\documentclass[12pt]{article}
\usepackage[utf8]{inputenc}

\usepackage[authoryear,longnamesfirst]{natbib}

\usepackage[margin=2.2cm]{geometry}
\usepackage[sups,scaled=0.91]{XCharter}
\usepackage[charter]{newtxmath}
\usepackage{inconsolata}

\usepackage[x11names]{xcolor}
\usepackage{pgfplots,tikz}
\pgfplotsset{compat=1.17}

\usepackage{amsthm}\usepackage{bm}\usepackage{mathrsfs}\usepackage{enumitem}
\usepackage{comment,setspace,xspace,subcaption}
\usepackage[colorlinks=true,linkcolor=black,anchorcolor=black,citecolor=DodgerBlue4,filecolor=black,menucolor=black,runcolor=black,urlcolor=DodgerBlue4]{hyperref}
\hypersetup{breaklinks=true}
\usepackage[capitalize]{cleveref}

\DeclareMathOperator*{\maxx}{\mathrm{max}}
\DeclareMathOperator*{\argmax}{\mathrm{arg}\!\maxx\limits}

\newcommand{\data}{\ensuremath{\mathit{data}}}

\usepackage{breakcites}

\newtheorem{assumption}{Assumption}\newtheorem{lemma}{Lemma}\newtheorem{theorem}{Theorem}\newtheorem{corollary}{Corollary}\newtheorem*{remark}{Remark}

\theoremstyle{remark}

\DeclareMathOperator*{\argmin}{arg\,min} 

\date{}

\title{
On the limitations of data-based price discrimination%
\thanks{Haitian Xie and Ying Zhu share the first authorship and are listed alphabetically. This paper supersedes a previously circulated draft by Xie and Zhu (\href{https://arxiv.org/pdf/2204.12723v1.pdf}{\tt https://arxiv.org/pdf/2204.12723v1.pdf}).
All three authors are grateful to the constructive comments from Co-Editor Rakesh Vohra and two anonymous reviewers at \emph{Theoretical Economics}, and three anonymous reviewers and the meta reviewer at \emph{ACM Economics and Computation}. The authors would also like to thank Dirk Bergemann, Songzi Du, Federico Echenique, Graham Elliott, Yannai Gonczarowski, Roger Gordon, Nima Haghpanah, Johannes Horner, Jonathan Libgober, Esfandiar Maasoumi, Maximilian Schaefer, Joel Sobel,
Karl Schlag, Larry Samuelson, Yixiao Sun, J. Miguel Villas-Boas, and Joel Watson
for valuable comments and discussions.
}}
\author{Haitian Xie\thanks{Assistant professor at the Department of Business Statistics and Econometrics, Guanghua School of Management, Peking University.
The author is grateful to the UC San Diego Department of Economics where this project was developed during his doctoral studies. Xie is supported by the Fundamental Research Funds for the Central Universities at Peking University. Email: \href{mailto:xht@gsm.pku.edu.cn}{\tt xht@gsm.pku.edu.cn}.}
\quad{}\quad{}Ying Zhu\thanks{Assistant professor at the Department of Economics, University of California San Diego%
. Email: \href{mailto:yiz012@ucsd.edu}{\tt yiz012@ucsd.edu}.
The author is grateful to the Society of Hellman Fellows at University
of California and the Cowles Foundation at Yale University, and also
thanks participants at her seminars.}
\quad{}\quad{}Denis Shishkin\thanks{Assistant professor at the Department of Economics, University of California San Diego%
. Email: \href{mailto:dshishkin@ucsd.edu}{\tt dshishkin@ucsd.edu}.}} 
\date{\today}

\begin{document}

\maketitle

\abstract{
\singlespacing


The classic third degree price discrimination (3PD) model requires the knowledge of the distribution of buyer valuations and the covariate to set the price conditioned on the covariate.
In terms of generating revenue, the classic result shows that 3PD is at least as good as uniform pricing.
What if the seller has to set a price based only on a sample of observations from the underlying distribution?
Is it still obvious that the seller should engage in 3PD?
This paper sheds light on these fundamental questions.
In particular, the comparison of the revenue performance between 3PD and uniform pricing is ambiguous overall when prices are set based on samples.
This finding is in the nature of statistical learning under uncertainty: a curse of dimensionality, but also other small sample complications.

\bigskip{}

\bigskip{}

}

\newpage
\section{Introduction}

In the past few decades, the advances in the theory of mechanism design have been followed by a tremendous interest in its practical applications. At the same time, classic theoretical models typically make strong assumptions about the designer's knowledge of the environment which may lead the optimal mechanism to be sensitive to the details of the environment (which is sometimes referred to as the Wilson critique).\footnote{In some cases, this leads to extreme or unrealistic results as in, e.g., \citet{cremer1988full}.} %

Third degree price discrimination (3PD) requires an observable covariate value associated with the buyer valuation. To set the price conditioned on the covariate, the classic pricing model requires the knowledge of the distribution of buyer valuations and the covariate. In terms of generating revenue, the classic result shows that 3PD is at least as good as uniform pricing. What if the seller has only partial information about those distributions? Is it still obvious that the seller should engage in 3PD? Setting the optimal price for each observed value of the covariate may not ``extrapolate'' well to the unobserved covariate values, and yield a lower expected revenue than a uniform price. On the other hand, too little discrimination underutilizes the information contained in the covariate about buyer valuations. This paper is concerned with how much information the seller will need in order to make 3PD generate more revenue. Suppose a unit demand buyer with a privately-known valuation $Y$ and a one-dimensional continuous covariate $X$ drawn from a joint distribution $F_{Y,X}$, that is unknown to the seller. The continuous covariate $X$ can be a single index or score that summarizes the relevant characteristics for pricing and marketing. \cite{hartmann2011identifying} provide examples where marketing firms use a one-dimensional continuous score function of customer characteristics, past response histories, and features of the zip code, and casinos use a one-dimensional continuous score referred to as the average daily win.


While our seller is ignorant of $F_{Y,X}$, he/she does have access to a random sample of i.i.d. $\{Y_{i},X_{i}\}_{i=1}^{n}$ drawn from $F_{Y,X}$. A natural strategy is to choose prices that optimize against the empirical distribution of $\{Y_{i},X_{i}\}_{i=1}^{n}$. The $K$-markets empirical revenue maximization (ERM) divides the covariate space into $K$ equal-length segments, and the optimal price based on the conditional empirical distribution for each segment is calculated. We show that when $K=\Theta(n^{1/4})$, the
$K$-markets ERM strategy generates an expected revenue converging to that of
the true distribution 3PD optimum at the rate $O(n^{-1/2})$. The $1$-market ERM strategy is simply the (uniform) ERM strategy, which we show generates a revenue converging to that of the true-distribution uniform optimum at the rate $O(n^{-2/3})$. The $K$-markets ERM is just one possible strategy and one may wonder if a more sophisticated strategy might provide faster convergence rates. In a sense, the answer is no. We show that these rates are asymptotically unimprovable for the worst case distributions of $(Y,X)$ subject to some mild smoothness conditions. In other words, to guarantee a revenue deficiency of $\delta$ uniformly over a class of distributions, the necessary condition for the sample size is that $n=\Omega(\delta^{-2})$ in the 3PD problem and $n=\Omega(\delta^{-3/2})$ in the uniform pricing problem.

For sufficiently small $\delta$, the $K$-markets ERM and the uniform ERM strategies are optimal on the growth requirements of the sample size, respectively; that is, $n=\Theta(\delta^{-2})$ in the 3PD problem and $n=\Theta(\delta^{-3/2})$ in the uniform pricing problem. To show this optimality result, we establish a lower bound for the revenue deficiency in \textit{any} data-based pricing strategy relative to the true-distribution optimal strategy in the worst case (by considering the supremum over a class of joint distributions, $F_{Y,X}$, subject to some mild smoothness assumptions). In particular, data-based uniform pricing strategies are algorithms that depend on $\{Y_{i}\}_{i=1}^{n}$ only, and the true-distribution optimal strategy corresponds to the optimal uniform pricing strategy derived from $F_{Y}$. Similarly, data-based 3PD strategies are algorithms that depend on $\{Y_{i},X_{i}\}_{i=1}^{n}$, and the true-distribution optimal strategy corresponds to the optimal 3PD strategy derived from $F_{Y,X}$. We show that the minimax revenue deficiency is $\Omega(n^{-2/3})$ and $\Omega(n^{-1/2})$ in the uniform and 3PD cases, respectively.

Our results highlight the following economic trade-off. When the seller has the access to a sample of i.i.d. $\{Y_{i},\,X_{i}\}_{i=1}^{n}$, she can choose the $K$-markets ERM strategy that exploits both $\{X_{i}\}_{i=1}^{n}$ and $\{Y_{i}\}_{i=1}^{n}$, or the uniform ERM strategy that ignores $\{X_{i}\}_{i=1}^{n}$ and exploits only $\{Y_{i}\}_{i=1}^{n}$. Inherently, the former is an algorithm trying to learn the $F_{Y,X}$-optimal pricing function $p(\cdot)$ while the latter is an algorithm trying to learn the $F_{Y}$-optimal (constant) pricing function. As a result of the curse from the \emph{extra} dimensionality, the former is more demanding in the sample size than the latter. On the other hand, in terms of generating revenue, the true-distribution optimal 3PD strategy is at least as good as the true-distribution optimal uniform pricing strategy. This trade-off suggests that, even if $X$ contains useful information about $Y$, the $K$-markets ERM strategy based on a random sample \emph{can} be revenue inferior to the uniform ERM strategy when the sample size $n$ is not large enough, and vice versa.

To verify these potential implications, we conduct several numerical studies. In particular, we calculate the revenues of the $K$-markets ERM and the uniform ERM strategies based on a real-world data set from eBay auctions and two simulated data sets. Our numerical results illustrate the aforementioned trade-off.
When the sample size is small, the uniform ERM strategy can generate higher expected revenue than the $K$-markets ERM strategy. As the sample size grows, the $K$-market ERM strategy (the uniform ERM strategy) gets closer to the true-distribution optimal 3PD strategy (respectively, the true-distribution optimal uniform pricing strategy). The slower rate of convergence in the revenue from the $K$-markets ERM strategy (in contrast to the faster rate of convergence in the revenue from the uniform ERM strategy) is dominated by the benefit of price discrimination (based on $F_{Y,X}$) over uniform pricing (based on $F_{Y}$). Consequently, the revenue of the $K$-markets ERM strategy overtakes that of the uniform ERM strategy when the sample size becomes sufficiently large and $X$ contains sufficient information about $Y$.

The key takeaways from this paper are summarized here. First, no sample-based 3PD strategy is able to escape from the curse of dimensionality, shown by our information theoretic lower bounds. Second, absent uncertainty regarding the underlying probability laws,
third-degree price discrimination is at least as good as uniform pricing
in generating revenue. In contrast, the comparison of the revenue
performance between the $K$-markets ERM and the uniform ERM strategies
is ambiguous overall. This finding is in the nature of statistical
learning under uncertainty: a curse of dimensionality, but
also other small sample complications.\footnote{Specifically, there exists a distribution $F_Y$ where the revenue of the uniform ERM strategy is worse with two observations than with one; see \cite{babaioff2018two}. We illustrate in Section \ref{sec:Discussions} that this seemingly counter-intuitive result highlights the difficulty of establishing general comparative results with very small sample size and sheds some light on the comparison of the revenue performance of the $K$-markets ERM strategy with $K=1$ vs $K=2$ in the case of $n=2$.} Empirical
revenue maximization is not free of these issues. Ultimately, this paper poses a challenging open question of whether there exist some $\underline{n}<\bar{n}<\infty$ such that for any $n\in [\underline{n},\bar{n}]$ and distribution in the class defined in this paper, the $K$-markets ERM strategy (for any $K>1$) is always revenue-inferior to the uniform ERM strategy.


\subsection{Related literature}
\paragraph{Complexity measures and information theoretic lower bounds.}
Information theoretic lower bounds and sample complexity are important notions in machine learning. 
Both aim to characterize learnability, i.e., how easy it is to learn an unknown object of interest (in our context, the true-distribution optimal 3PD strategy) from data where the uncertainty arises. Sample complexity derives the rate at which the sample size needs to grow to guarantee a desired learning accuracy. Information theoretic lower bound derives a lower bound as a function of the sample size on the learning error (in our context, the revenue deficiency) in the worst case. Sample complexity and information theoretic lower bounds are intrinsically tied to the complexity or size of the underlying function class of interest. Vapnik-Chervonenkis (VC) dimensions, shattering dimensions, and metric entropy (such as the cardinality of packing sets) are popular measures of complexity in machine learning. There have been a number of innovative applications of VC dimensions or shattering dimensions in economic theory and algorithmic economics. Together with the Probably Approximately
Correct (PAC) framework, they are used to study the complexity of the classes of demand and utility functions \citep{beigman2006learning,balcan2014learning}, k-demand buyer's valuation \citep{zhang2020learning}, theories of choices \citep{basu2020falsifiability}, preference functions \citep{chambers2021recovering,chambers2023recovering}, as well as the resulting learnability from data. VC dimension is useful for deriving sample complexity bounds concerning discrete function sets and finite dimensional vector spaces, and shattering dimension is useful for certain real functions.

From the theory of machine learning, when a class has infinite VC or shattering dimensions, this class is not PAC learnable. For example, a collection of sinusoids 
have subgraphs with infinite VC dimension. The max-min expected utility model with at least three states of the world has infinite VC dimension \citep{basu2020falsifiability}. The class of demand functions has infinite shattering dimension \citep{beigman2006learning}. Nonetheless, the notion of "learnability" can be generalized using a different type of complexity analysis that gives rise to our information theoretic lower bound in the 3PD problem. This type of analysis is built upon the notion of packing sets, along with tools from information theory. In particular, packing sets are useful for studying classes with an infinite number of elements (see \cite{kolmogorov1959varepsilon} and \cite{wainwright2019high}). This is the case for our 3PD problem as we try to learn an optimal pricing function of the covariate (an infinitely-dimensional parameter) and bound the deficiency in the \emph{expected} revenue, which concerns the entire pricing function at all covariate values. 


\paragraph{Prior-independent mechanism design.}
Most of the classic monopoly pricing literature assumes a known distribution of valuations (and covariates).\footnote{See also \cite{segal2003optimal} for a study of optimal multi-unit auctions where the seller has a probabilistic belief about the valuation distribution of the i.i.d. buyers.}
More recently, some papers \citep[e.g., those surveyed in][]{carrollRobustnessMechanismDesign2019} studied ``prior''-independent mechanism design.\footnote{Here, ``prior'' distribution refers to the seller's prior belief about buyers' valuations and is often taken to be the true distribution.}
The main focus of that literature is on deriving a robustly optimal mechanism in the absence of both ``prior'' and data.
In particular, \cite{BergemannSchlag2008,BergemannSchlag2011} derive the minimax-regret uniform pricing strategy in closed form; that is, the strategy that guarantees the smallest deficiency in revenue relative to the known distribution case. Like \cite{BergemannSchlag2008,BergemannSchlag2011}, we study the revenue deficiencies, but in contrast, we assume the availability of data and focus on the (inevitable) information-theoretic limitations of any \textit{data-based} pricing strategies and the achievability of the limitation.

This paper is inspired by the literature that studies approximately optimal ``prior''-independent mechanism design, in particular monopoly pricing with a single buyer.\footnote{There is a less related literature that studies optimal auctions; see, e.g., \cite{cole2014sample,dhangwatnotai2015revenue,fu2015randomization,guo2019settling,fu2021full}.} 
This literature assumes that the seller has access to a random sample of i.i.d. $\{Y_{i}\}_{i=1}^{n}$ drawn from $F_{Y}$ and proposes variants of the uniform ERM strategy to derive the revenue guarantee in relation to that from the true-distribution optimal uniform pricing strategy.
There are two types of analyses in this literature.
The first one focuses on the guarantees for the specific case of $n=1$ or $n=2$ \citep{babaioff2018two,allouah2023optimal}.
The second one \citep[e.g.,][]{huang2018making} establishes ``sample complexity bounds'' such that the uniform ERM variants achieve a $\left(1-\epsilon\right)$ fraction guarantee when the sample size grows at a rate depending on $\epsilon$, and also derives the rate at which the sample size needs to grow (as a function of $\epsilon$) for any data-based uniform pricing strategies to obtain a given $\left(1-\epsilon\right)$ fraction guarantee. \cite{allouah2022pricing} involves both types of analyses.

In this paper, we ask the related question, how fast the revenue deficiency decays as a function of $n$, and provide an answer using information-theoretic lower bounds (independent of algorithms) and upper bounds with respect to specific algorithms in the worst case scenarios.\footnote{A large literature studies data-based auctions by focusing on guarantees for revenue deficiencies (instead of fractions), such as how the revenues from the data-based strategies converge in probability to the true-distribution benchmark, e.g., \cite{baliga2003market,goldberg2006competitive,gonccalves2020statistical}. This line of work does not consider the optimal rates of convergence or optimal sample size requirements.} The main difference with the majority of the data-based literature is that, we study third-degree price discrimination (3PD) with a continuous covariate and compare the revenue performance of data-based 3PD and uniform pricing strategies. 

To understand why the 3PD problem in our context is more challenging than the uniform pricing problem, note that fundamentally, the latter tries to learn the constant optimal pricing function (a scalar parameter) while the former tries to learn an optimal pricing function of the covariate (an infinitely-dimensional parameter), where the deficiency in the \textit{expected} revenue concerns the entire pricing function at all covariate values. 
Our framework allows us to tackle several challenging aspects of the 3PD problem, which might be difficult to analyze with the toolkit in the existing pricing literature.
We describe one example below. 

Somewhat related, \cite{devanur2016sample} studies sample complexity of optimal pricing with ``side information''.
In their ``signals model'' (Sections 5.1 and 5.3), there is a covariate (signal) $X \in [0,1]$, and the seller can condition the data-based reserve price on the covariate.
For the single-buyer case (which would correspond to our 3PD problem), they derive upper and lower sample complexity bounds.
Importantly, they assume that the true joint distribution $F_{Y,X}$ has the following property: larger values of $X$ are associated with larger values of $Y$ in the sense of first-order stochastic dominance of conditional distributions.
In contrast, our 3PD setup imposes no assumptions about the relationship between the covariate $X$ and the valuation $Y$; meanwhile, our proposed $K$-market ERM strategy learns the relationship from the data. 
Moreover, our $K$-market ERM strategy attains the optimal rate of convergence in revenue deficiency (as described before), while the upper and lower bounds in \cite{devanur2016sample} have different rates, and hence, the optimal sample size requirement is unclear.

\section{Setup}
\label{sec:the setup}

The seller is selling an
item to a buyer. Let $Y\in[0,1]$ be the valuation (i.e., willingness
to pay) of the buyer, and $X$ the covariate (such as a characteristic) associated with the buyer. The joint distribution of $(Y,X)$
is denoted by $F_{Y,X}$. We assume that $X$ is supported on a bounded
interval, and without loss of generality, we take the interval to
be $[0,1]$.\footnote{The assumption that $Y,X \in [0,1]$ is made merely for simplicity. First of all, our results in Sections \ref{sec:upper} and \ref{sec:lower} hold for general bounded supports. Second, the precise knowledge of the support boundaries is unnecessary because they can be readily estimated using extremum order statistics. The estimator converges at a superconsistent rate of $n^{-1}$ \citep[see, e.g.,][]{hirano2003asymptotic}, significantly faster than the convergence of revenue deficiency that we show in Section \ref{sec:upper}. Therefore, in our analysis, the estimation error resulting from the unknown support is negligible. We are grateful to a referee for raising this discussion.\label{footnote:support}}

Given a covariate value, the seller wants to set a price
according to a mapping from the covariate to a set of prices. We use
$\mathcal{D}$ to denote the set of all pricing functions: 
\begin{align*}
\mathcal{D} & \equiv\{p\colon[0,1]\rightarrow[0,1],\text{ measurable}\}.
\end{align*}
For a generic pricing strategy $p\in\mathcal{D}$, the price depends
on the covariate value $x$. This scheme falls in the realm of third-degree
price discrimination (3PD). Uniform pricing can be viewed as a special
case where the price is the same for all covariate values. We use
$\mathcal{U}$ to denote the set of all uniform pricing functions:
\begin{align*}
\mathcal{U} & \equiv\{p\in\mathcal{D}\colon p\text{ is a constant function}\}.
\end{align*}
To lighten the notation, we express $p\in\mathcal{U}$ as a scalar
rather than a function for the uniform pricing problem.

Let $F_{Y|X}$ be the conditional CDF and $f_{X}$ the marginal
density function. Given a price $y\in[0,1]$ and a covariate value $x\in[0,1]$,
there are $1-F_{Y|X}(p|x)$ buyers whose valuation is above
the price. The revenue generated from these buyers is 
\begin{align}
r(y,x,F_{Y,X})\equiv(1-F_{Y|X}(y|x))y, \label{eq:revenue opt}
\end{align}
and the \emph{expected} revenue for a pricing function $p$ is 
\begin{align*}
R(p,F_{Y,X})\equiv\int_{0}^{1}r(p(x),x,F_{Y,X})f_{X}(x)dx.
\end{align*}
In various places of the rest of the paper, we will slightly abuse the notation and denote $r(p,x)\equiv r(p(x),x)$ when $p$ is a pricing function and also write $r(y,x) = r(y,x,F_{Y,X})$ for brevity when $F_{Y,X}$ is clear from the context.
In the special case where the pricing strategy is uniform (i.e., $p\in\mathcal{U}$),
the revenue only depends on the marginal distribution $F_{Y}$: 
\begin{equation*}
R(p,F_{Y,X}) 
=p\mathbb{P}(Y\geq p)
 =p(1-F_{Y}(p)),p\in\mathcal{U}.
\end{equation*}
The true-distribution optimal 3PD strategy $p_{D}^{*}$ is the one
that maximizes the revenue: 
\begin{align*}
R(p_{D}^{*},F_{Y,X})=\sup_{p\in\mathcal{D}}\int_{0}^{1}r(p(x),x,F_{Y,X})f_{X}(x)dx.
\end{align*}
In a similar fashion, we denote $p_{U}^{*}$ as the true-distribution
optimal uniform pricing strategy such that 
\begin{align*}
R(p_{U}^{*},F_{Y})=R(p_{U}^{*},F_{Y,X})=\sup_{p\in\mathcal{U}}p(1-F_{Y}(p)).
\end{align*}
Note that $p_{D}^{*}$ depends on $F_{Y,X}$ and $p_{U}^{*}$ depends
on $F_{Y}$.

In terms of generating revenue, the classic pricing theory shows
that 3PD is at least as good as uniform pricing when the joint distribution
$F_{Y,X}$ is known to the seller. In this case, we can solve analytically or numerically
for the optimal pricing strategies $p_{D}^{*}$ and $p_{U}^{*}$.
Since $\mathcal{U}$ is contained in $\mathcal{D}$, $p_{D}^{*}$
must achieve a (weakly) better revenue than $p_{U}^{*}$. Intuitively,
when $Y$ is correlated with $X$, $p_{D}^{*}$ utilizes the information
in $X$.

Now suppose that the seller knows neither $F_{Y,X}$ nor $F_{Y}$,
but instead observes a random sample of $\textit{data}\equiv\{(Y_{i},X_{i}),1\leq i\leq n\}$
drawn from $F_{Y,X}$, or $\textit{data}_{Y}\equiv\{Y_{i},1\leq i\leq n\}$
from $F_{Y}$, and wants to construct a pricing strategy based
on the sample.
The following assumption is used throughout this paper.

\begin{assumption} \label{assumption:i.i.d.}
    $\textit{data}$ and $\textit{data}_{Y}$ consist of i.i.d. draws from $F_{Y,X}$ and $F_{Y}$, respectively.
\end{assumption}
The following assumption is used to establish the results concerning our 3PD problem. 
Instead of a single known joint distribution $F_{Y,X}$, there is a class $\mathcal{F}$ of unknown distributions which are deemed possible and our data-based pricing strategies can be evaluated within this class.
The functions in $\mathcal{F}$ satisfy several smoothness and regularity conditions stated below.

\begin{assumption} \label{def:pd}

Any distribution function in the set $\mathcal{F}$ satisfies the following conditions.

\begin{enumerate}[label = (\roman*)]
\item\label{as:lipschitz} (Lipschitz continuity) There exists $C_{0}\in(0,\infty)$ such that,
for any $y,y',x\in[0,1]$, the conditional density $f_{Y|X}$ satisfies
\begin{align*}
|f_{Y|X}(y|x)-f_{Y|X}(y'|x)|\leq C_{0}|y-y'|.
\end{align*}
\item\label{as:concave} (Strong concavity) 
 There exists $C^{*}>0$ such that the revenue function $r(y,x)\equiv y(1-F_{Y|X}(y|x))$
is strictly concave with second-order derivative 
\begin{align}
-2f_{Y|X}(y|x)-y\frac{\partial}{\partial y}f_{Y|X}(y|x)\leq-C^{*},\text{ a.e.}\label{eq:2nd derivative}
\end{align}
\item\label{as:interior} (Interior solution) For each $x\in[0,1]$, the optimal price is an
interior solution; that is, $p_{D}^{*}(x;F_{Y,X})\in(0,1)$. 
\item\label{as:diff} (Differentiability) The conditional distribution function $f_{Y|X}(y|x)$
is continuously differentiable in $(x,y)$ in a neighborhood of the
curve $\{(x,p_{D}^{*}(x;F_{Y,X})):x\in[0,1]\}$. 
\item\label{as:bounded} (Boundedness) The functions 
\begin{align}
\left|2f_{Y|X}(y|x)+y\frac{\partial}{\partial y}f_{Y|X}(y|x)\right|\label{eq:Boundedness_1}\\
\text{ and }\left|\frac{\partial}{\partial x}F_{Y|X}(y|x)+y\frac{\partial}{\partial x}f_{Y|X}(y|x)\right|\label{eq:Boundedness_2}
\end{align}
are bounded from above by $\overline{C}\in\left(0,\,\infty\right)$
a.e.
\item\label{as:marg} (Marginal density) The marginal density $f_{X}$ is 
bounded from above
by $\overline{C}'\in\left(0,\,\infty\right)$ and 
bounded away
from zero; that is, $f_{X}\geq\underline{C}>0$. 
\end{enumerate}
\end{assumption}

Part \ref{as:lipschitz} requires the conditional density function to be sufficiently smooth.
The partial derivative $\frac{\partial}{\partial y}f_{Y|X}(y|x)$ is well defined almost everywhere because $f_{Y|X}$ is Lipschitz continuous and hence absolutely continuous. Part \ref{as:interior} ensures that the first-order condition holds for the optimal price.
Part \ref{as:diff} ensures that the optimal pricing function $p_{D}^{*}(x;F_{Y,X})$ is sufficiently smooth in $x$.
Part \ref{as:bounded} requires the partial derivatives of the revenue to be bounded.
Part \ref{as:marg} ensures that the covariate does not take vanishing or dominating values.

Under part \ref{as:concave}, the optimal price is well defined. Part \ref{as:concave} is a standard assumption in the optimal auctions/pricing literature also known as regularity \citep{myerson1981optimal}, which is a so-called ``strong concavity" condition from machine learning theory. It is well known that any distribution $F$ with the monotone hazard rate satisfies regularity.


Analogously, the following assumption is used to establish the results for the uniform pricing problem which concerns a class $\mathcal{F}^{U}$ of unknown marginal distributions that are deemed possible.


\begin{assumption} \label{def:uniform} 
Let $\mathcal{F}^{U}$ be the set of marginal distributions such that any $F_{Y}\in \mathcal{F}^{U}$ satisfies parts (i), (ii), and (v)(\ref{eq:Boundedness_1})
of \cref{def:pd} with $f_{Y|X}(y|x)$ replaced by $f_{Y}(y)$.
Moreover, the optimal price is an interior solution; that is, $p_{U}^{*}(F_{Y})\in(0,1)$.
The distribution function $f_{Y}(y)$ is continuously differentiable in $y$ in a neighborhood of $p_{U}^{*}(F_{Y})$. 

\end{assumption}

\begin{remark}
    By defining $\mathcal{F}^{U}$ in the way above, note that the marginal distribution associated with any joint distribution satisfying (i), (ii) and (v)(\ref{eq:Boundedness_1}) of Assumption \ref{def:pd} satisfies the counterpart conditions in Assumption \ref{def:uniform}.
\end{remark}

\paragraph{Notation.}

For functions $f(n)$ and $g(n)$, we write $f(n)\gtrsim g(n)$ to mean that $f(n)=\Omega(g(n))$. Similarly, we write $f(n)\lesssim g(n)$ to
mean that $f(n)=O(g(n))$. The notation
$f(n)\asymp g(n)$ means that $f(n)=\Theta(g(n))$; that is, $f(n)=\Omega(g(n))$ and $f(n)=O(g(n))$. As a general rule for this paper, the various
$c$ and $C$ constants denote positive universal constants that are
independent of the sample size $n$, and may vary from place to place.
For functions $f$ and $g$, the unweighted $L_{2}$ norm ($L_{2}$ as the short
form) $\lVert f-g\rVert_{2}\equiv\left(\int_{0}^{1}\left[f\left(x\right)-g\left(x\right)\right]^{2}dx\right)^{\frac{1}{2}}$.

\section{The $K$-markets ERM strategy}

\label{sec:upper}

In this section, we
propose the $K$-markets ERM strategy, and compare its revenue with that of the true-distribution optimal 3PD strategy.
In particular, we provide upper bounds for the pointwise and expected revenue deficiency as a function of $n$. We also compare the revenue of the $1$-market (uniform) ERM strategy with that of the true-distribution uniform optimum, and provide an upper bound on the revenue deficiency.


\subsection{Price discrimination}\label{sec:upper-1}

We propose the \emph{``$K$-markets''} ERM strategy
for the data-based 3PD problem with a continuous covariate:

\begin{enumerate}
\item Divide the individuals into $K(\equiv K_{n})$ markets by splitting
the covariate space $[0,1]$ into $K$ equally spaced intervals 
\begin{align*}
I_{k}\equiv[(k-1)/K,k/K],k=1,\ldots,K.
\end{align*}
\item For each market $I_{k}$, based on the empirical distribution of $\{Y_{i}\colon X_{i}\in I_{k}\}$,
\begin{equation}
\hat{F}_{k}(y)=\frac{1}{n_{k}}\sum_{i\colon X_{i}\in I_{k}}\mathbf{1}\{Y_{i}\leq y,X_{i}\in I_{k}\} \label{eq:n_k}
\end{equation}
where $n_{k}$ is the cardinality of $\{i\colon X_{i}\in I_{k}\}$, solve
for the optimal price $\hat{p}_{D,k}$ as follows,
\begin{align*}
\hat{p}_{D,k}\equiv\argmax_{p\in [0,1]}p(1-\hat{F}_{k}(p)).
\end{align*}
The resulting pricing function is a piece-wise function 
\begin{align*}
\hat{p}_D(x;\textit{data})=\hat{p}_{D,k},x\in I_{k}.
\end{align*}
If the $k$th market does not contain any observation,
then simply choose $\hat{p}_{D,k}$ to be any arbitrary number in
$[0,1]$. Doing so has no impact on the asymptotic guarantee implied
by the following theorem. For practical implementation, the desired
choice may change from context to context, depending on the seller's
specific knowledge about a buyer, and the related analysis would be
beyond the scope of this paper. 
\end{enumerate}

\begin{theorem} \label{thm:upper-bound-pd} Suppose Assumptions \ref{assumption:i.i.d.}
and \ref{def:pd} hold. There exists a positive universal
constant $c_{1}\in(0,\infty)$ such that the following results hold.\footnote{For example, the constant $c_1=1$ when $X \sim U[0,1]$.}

\begin{enumerate}[label = (\roman*)]
\item At a given covariate value $x_{0}$, the revenue generated by the
$K$-markets ERM strategy $\hat{p}_{D}$ satisfies 
\begin{align*}
\sup_{F_{Y,X}\in\mathcal{F}}\left(r(p_{D}^{*},x_{0})-\mathbb{E}_{F_{Y,X}}\left[r(\hat{p}_{D}(\textit{data}),x_{0})\right]\right)\lesssim1/K^{2}+(K/n)^{2/3}\\
{+\exp\left(-\frac{nc_{1}^{2}}{8K^{2}}+\log K\right)},\quad x_{0}\in I_{k},
\end{align*}
where the expectation $\mathbb{E}_{F_{Y,X}}$ is taken with respect
to $\textit{data}\sim F_{Y,X}$ and $K$ satisfies
$\frac{c_{1}}{K}\leq\frac{1}{2}$; moreover, 
\[
(K/n)^{2/3}+1/K^{2}\asymp n^{-1/2}\textrm{ when }K\asymp n^{1/4},
\]
in which case,
\[
\sup_{F_{Y,X}\in\mathcal{F}}\left(r(p_{D}^{*},x_{0})-\mathbb{E}_{F_{Y,X}}\left[r(\hat{p}_{D}(\textit{data}),x_{0})\right]\right)\lesssim n^{-1/2}.
\]

\item The \textit{expected} revenue generated by the $K$-markets ERM strategy
$\hat{p}_{D}$ satisfies 
\begin{align*}
\sup_{F_{Y,X}\in\mathcal{F}}\left(R(p_{D}^{*},F_{Y,X})-\mathbb{E}_{F_{Y,X}}\left[R(\hat{p}_{D}(data),F_{Y,X})\right]\right)\lesssim1/K^{2}+(K/n)^{2/3}\\
{+\exp\left(-\frac{nc_{1}^{2}}{8K^{2}}+\log K\right)}
\end{align*}
where the expectation $\mathbb{E}_{F_{Y,X}}$ is taken
with respect to $\textit{data}\sim F_{Y,X}$ and $K$ satisfies $\frac{c_{1}}{K}\leq\frac{1}{2}$;
moreover, 
\[
(K/n)^{2/3}+1/K^{2}\asymp n^{-1/2}\textrm{ when }K\asymp n^{1/4},
\]
in which case,
\[
\sup_{F_{Y,X}\in\mathcal{F}}\left(R(p_{D}^{*},F_{Y,X})-\mathbb{E}_{F_{Y,X}}\left[R(\hat{p}_{D}(data),F_{Y,X})\right]\right)\lesssim n^{-1/2}.
\]

\end{enumerate}
\end{theorem}

\begin{remark} The term $\exp\left(-\frac{nc_{1}^{2}}{8K^{2}}+\log K\right)$ is technical and comes from a binomial tail bound on $n_k$ in (\ref{eq:n_k}); see (\ref{eq: A_q event}) and the following derivation in the appendix for more detail. Suppose $8K^{2}=n^{1-c}c_{1}^{2}$
with $c\in(0,1)$ so that $\frac{nc_{1}^{2}}{8K^{2}}=n^{c}$ (for
example, $c=\frac{1}{2}$ which gives $K\asymp n^{1/4}$ as in the
theorem above). Then, there exists some positive universal constant
$c_{0}\in(0,\infty)$ such that
\[
\exp\left(-\frac{nc_{1}^{2}}{8K^{2}}+\log K\right)=\exp(-c_{0}n^{c})\:\textrm{ as }n\rightarrow\infty.
\]
In this case, note that $\exp(-c_{0}n^{c})=o\left((K/n)^{2/3}\right)$
and the term $\exp\left(-\frac{nc_{1}^{2}}{8K^{2}}+\log K\right)$
can be dropped from the bounds in Theorem \ref{thm:upper-bound-pd}. 

\end{remark}

Note that having an upper bound on the supremum of the revenue deficiency immediately implies that this upper bound holds for every distribution $F_{Y,X} \in \mathcal{F}$.
Moreover, the revenue of the $K$-markets ERM strategy is guaranteed to have a convergence rate no greater than the provided upper bound, in particular $n^{-1/2}$ when $K \asymp n^{1/4}$.

The interpretation of our results is as follows.
The deficiency in revenues comes from two sources. The
first part $(K/n)^{2/3}$ is related to the ``variance'', which is due to the randomness of the sample, making $\hat{F}_{k}(\cdot)$ different from its expectation.
The second part $1/K^{2}$ is related to the approximation error due to the fact that we set
the same price for all covariate values in the market $I_{k}$.
Note that more discrimination (larger $K$) increases the ``variance'' but reduces the approximation error, and selecting $K \asymp n^{1/4}$ minimizes the upper bound on revenue deficiency.

To show $(K/n)^{2/3}$, we use a peeling argument and other tools from empirical process theory \citep{alexander1987rates,wellner1996,van_de_geer2000empirical}. Even though this toolkit is widely used in mathematical statistics and theoretical machine learning, to our knowledge, it has not been introduced to the data-based pricing literature. Showing $1/K^{2}$ requires controlling $|\tilde{p}_{k}-p_{D}^{*}(x_{0})|$, where $\tilde{p}_{k}\equiv\argmax_{p\in[0,1]}p\mathbb{P}(Y>p,X\in I_{k})$ and $x_{0}\in I_{k}$. Using the implicit function theorem, we show that, (i) $p_{D}^{*}(x)$ is Lipschitz continuous
on $[0,1]$, and (ii) $\tilde{p}_{k}$ is a weighted average
of $p_{D}^{*}(x),x\in I_{k}$. These facts imply that $|\tilde{p}_{k}-p_{D}^{*}(x_{0})|^{s}\lesssim {1}/K^{s}$ for any fixed $s\geq1$.


\subsection{Uniform pricing\label{sec:upper-2}}

Based on the empirical distribution of $\{Y_{i}\}_{i=1}^{n}$
\[
\hat{F}(y)=\frac{1}{n}\sum_{i=1}^{n}\mathbf{1}\{Y_{i}\leq y\},
\]
the uniform ERM strategy simply solves for the optimal
price $\hat{p}_{U}$ as follows:
\begin{align*}
\hat{p}_{U}(\textit{data}_{Y})\equiv\argmax_{p\in[0,1]}p(1-\hat{F}(p)).
\end{align*}
We have the following result as a corollary of Theorem \ref{thm:upper-bound-pd}.
\begin{corollary} \label{thm:upper-bound-unif} Let Assumptions \ref{assumption:i.i.d.} and \ref{def:uniform} hold. The revenue generated
by $\hat{p}_{U}$ satisfies 
\begin{align*}
\sup_{F_{Y}\in\mathcal{F}^{U}}\left(R(p_{U}^{*},F_{Y})-\mathbb{E}_{F_{Y}}\left[R(\hat{p}_{U}(\textit{data}_{Y}),F_{Y})\right]\right)\lesssim n^{-2/3}
\end{align*}
where the expectation $\mathbb{E}_{F_{Y}}$ is taken with respect
to $\textit{data}_{Y}\sim F_{Y}$.
\end{corollary}

The 3PD ERM problem with a continuous covariate is more delicate than the uniform ERM problem.
The latter does not involve a (continuous) covariate and hence incurs no approximation error. Contrasting Corollary \ref{thm:upper-bound-unif} with Theorem \ref{thm:upper-bound-pd}, one can see that the only source of revenue deficiency in the uniform ERM strategy comes from the ``variance''. 

\subsection{Welfare analysis}

From the perspective of a policy maker, it is also of interest to study
the welfare under the specific pricing strategies in Sections \ref{sec:upper-1}
and \ref{sec:upper-2}. In this section, we derive the rate at which
the welfare generated by these data-based pricing strategies converges
to the welfare generated by their respective true-distribution optimal strategies. 

We assume that there is no production cost for the seller, and there is no utility for the seller if the item is not sold. 
These assumptions are typically imposed in a benchmark
model in the auction and pricing literature. For any pricing strategy
$p\in\mathcal{D}$, its welfare can be written as 
\begin{align*}
W(p,F_{Y,X})\equiv\mathbb{E}_{F_{Y,X}}[Y\mathbf{1}\{Y>p(X)\}].
\end{align*}

\begin{theorem} \label{thm:welfare} \ 

\begin{enumerate}[label = (\roman*)]
\item Let Assumptions \ref{assumption:i.i.d.} and \ref{def:pd} hold.
Take $K\asymp n^{1/4}$ in the \emph{``$K$-markets''} ERM strategy. Then 
\begin{align*}
\sup_{F_{Y,X}\in\mathcal{F}}\mathbb{E}_{F_{Y,X}}|W(\hat{p}_{D}(\textit{data}),F_{Y,X})-W(p_{D}^{*},F_{Y,X})|\lesssim n^{-1/4}.
\end{align*}
\item Let Assumptions \ref{assumption:i.i.d.} and \ref{def:uniform} hold.
Then 
\begin{align*}
\sup_{F_{Y}\in\mathcal{F}^{U}}\mathbb{E}_{F_{Y}}|W(\hat{p}_{U}(\textit{data}_{Y}),F_{Y})-W(p_{U}^{*},F_{Y})|\lesssim n^{-1/3}.
\end{align*}
\end{enumerate}
\end{theorem}



\section{Information-theoretic limitation of data-based pricing}\label{sec:lower}

The revenue deficiency
in the $K$-markets ERM strategy and uniform ERM strategy in Section \ref{sec:upper} is $O\left(n^{-\frac{2}{2+2}}\right)$
and $O\left(n^{-\frac{2}{2+1}}\right)$, respectively. Note the ``$2$''
and ``$1$'' in the second terms of the denominators of the exponents
in these upper bounds, where the ``$2-1=1$'' difference is a result
of the extra dimension from the covariate $X$ in the 3PD problem.
Without any lower bounds, the upper bounds alone are unable to confirm
that the curse of the extra dimensionality necessarily exists and is
unimprovable.

In this section, we establish lower bounds to show that \textit{no 3PD strategy
is able to escape the curse of the extra dimensionality} and hence the $K$-markets ERM strategy is not an exception. Our lower bounds also conclude the optimality of the convergence rates $n^{-1/2}$ and $n^{-2/3}$ from Section \ref{sec:upper} within the respective realms of 3PD and uniform pricing. Therefore, the dependence of the extra dimension due to $X$ in our 3PD problem cannot be improved. As discussed in the introduction, rate optimality speaks to the optimality or efficiency of the growth requirement of the sample size.


For the lower bounds, it makes little sense to consider a framework recommending the data-based
pricing strategies that are only good for a single distribution. For
any \textit{fixed} joint distribution $F_{Y,X}$, there is always
a trivial data-based pricing strategy: simply ignore the data and
select the optimal pricing scheme given $F_{Y,X}$. For this particular
distribution, the revenue deficiency is zero. However, such a pricing
strategy may perform poorly under other distributions of $(Y,X)$.
One solution to circumvent this issue is to compute the worst revenue deficiency over the class $\mathcal{F}$ of possible distributions.

To be specific, we consider the minimax difference in the revenues
at a given covariate value $x_{0}$ for 3PD,
\begin{align*}
\mathscr{R}_{n}^{D}(x_{0};\mathcal{F})\equiv\inf_{\check{p}_{D} \in \check{\mathcal{D}}}\sup_{F_{Y,X}\in\mathcal{F}}\left(r(p_{D}^{*},x_{0},F_{Y,X})-\mathbb{E}_{F_{Y,X}}\left[r(\check{p}_{D}(\textit{data}),x_{0},F_{Y,X})\right]\right),
\end{align*}
and the minimax difference in the \textit{expected} revenues for 3PD,
\begin{align*}
\mathscr{R}_{n}^{D}(\mathcal{F})\equiv\inf_{\check{p}_{D}\in\check{\mathcal{D}}}\sup_{F_{Y,X}\in\mathcal{F}}\left(R(p_{D}^{*},F_{Y,X})-\mathbb{E}_{F_{Y,X}}[R(\check{p}_{D}(\textit{data}),F_{Y,X})]\right),
\end{align*}
where the expectation $\mathbb{E}_{F_{Y,X}}$ is taken with respect
to $\textit{data}\sim F_{Y,X}$ and $R(\cdot,\cdot)$ is defined in Section \ref{sec:the setup}. 
In the definitions above, $\check{p}_{D}(\textit{data})$ is a pricing function in $\mathcal{D}$ and $\check{p}_{D}(x_{0};\textit{data})$ corresponds to its value at a covariate $x_{0}\in[0,1]$; moreover, $\check{\mathcal{D}}$ is the set of \textit{all data-based} 3PD functions $\check{p}_{D}$.

Similarly, for uniform pricing, we consider 
\begin{align*}
\mathscr{R}_{n}^{U}(\mathcal{F}^{U})\equiv\inf_{\check{p}_{U}\in\check{\mathcal{U}}}\sup_{F_{Y}\in\mathcal{F}^{U}}\left(R(p_{U}^{*},F_{Y})-\mathbb{E}_{F_{Y}}[R(\check{p}_{U}(\textit{data}_{Y}),F_{Y})]\right),
\end{align*}
where the expectation $\mathbb{E}_{F_{Y}}$ is taken with respect
to $\textit{data}_{Y}\sim F_{Y}$. 
In the definition above, $\check{p}_{U}(\textit{data}_{Y})$ is a uniform pricing function in $\mathcal{U}$ and $\check{p}_{U}(x_{0};\textit{data}_Y)$ corresponds to its value at a covariate $x_{0}\in[0,1]$; moreover, $\check{\mathcal{U}}$ is the set of \textit{all data-based} uniform pricing functions $\check{p}_{U}$.

In what follows, we derive a lower bound for $\mathscr{R}_{n}^{D}(x_{0};\mathcal{F})$,
$\mathscr{R}_{n}^{D}(\mathcal{F})$ and $\mathscr{R}_{n}^{U}(\mathcal{F}^{U})$,
respectively. These lower bounds are algorithm independent and reveal
the fundamental information-theoretic limitation of data-based pricing strategies. 

\subsection{Price discrimination} \label{sec: 3PD lower bound}



The first theorem presents a lower bound  for the revenue difference at a given covariate value $x_{0}$,
between \textit{any} data-based 3PD strategy and the true-distribution
optimal 3PD strategy under the worst-case distribution by taking the supremum
over $\mathcal{F}$.

\begin{theorem}[Lower bounds for 3PD, deficiency in pointwise revenue] \label{thm:lower-bound-pd-1}
Let \cref{assumption:i.i.d.} hold. For any $\mathcal{F}$ satisfying \cref{def:pd} with $C^{*}\in(0,2)$ in (\ref{eq:2nd derivative}), the minimax difference
in the revenues at a given covariate value $x_{0}$ is bounded from
below as
\begin{align*}
\mathscr{R}_{n}^{D}(x_{0};\mathcal{F})\gtrsim n^{-1/2},\quad x_{0}\in(0,1),
\end{align*}
if $x_{0}n^{1/4}\geq c'$ and $(1-x_{0})n^{1/4}\geq c''$ for
some positive universal constants $c'$ and $c''$ (independent
of $n$ and $x_{0}$).

\end{theorem}

The second theorem presents a lower bound  for the difference in \textit{expected} revenues between \textit{any}
data-based 3PD strategy and the true-distribution optimal 3PD
strategy under the worst-case distribution by taking the supremum over
$\mathcal{F}$.

\begin{theorem}[Lower bounds for 3PD, deficiency in expected revenue] \label{thm:lower-bound-pd} Let \cref{assumption:i.i.d.} hold. For any $\mathcal{F}$ satisfying \cref{def:pd} with $C^{*}\in(0,2)$ in (\ref{eq:2nd derivative}), the minimax difference
in the \textit{expected} revenues is bounded from below as 
\begin{align*}
\mathscr{R}_{n}^{D}(\mathcal{F})\gtrsim n^{-1/2}.
\end{align*}

\end{theorem}

\begin{remark} By requiring $C^{*}\in(0,2)$ in the theorems above,
we allow $r(y,x)$ associated with an $f_{Y|X}$ to have a second
derivative bounded from above by a number smaller than or equal to
$-2$. To motivate the use of $C^{*}\in(0,2)$, suppose $f_{Y|X}=f_{Y}$
(that is, the valuation and covariate are independent of each other)
and $f_{Y}$ is the uniform distribution on $[0,1]$, $U[0,1]$. In
this case, the revenue function equals $R(y)=y(1-y)$, which is twice-differentiable
with second-order derivative $R''(y)=-2$ for any $y\in[0,1]$.
In our proof for the lower bounds, $U[0,1]$ is used as the benchmark
distribution. 

\end{remark}

Theorems \ref{thm:lower-bound-pd-1} and \ref{thm:lower-bound-pd}
state that, there is an inevitable deficiency, $\Omega(n^{-1/2})$,
in the revenue from \textit{any} data-based 3PD strategy relative to
the revenue from the true-distribution optimal 3PD strategy in the worst
case by taking the supremum over $\mathcal{F}$. 


Recalling Theorem \ref{thm:upper-bound-pd} on the convergence rate $O(n^{-1/2})$ of the revenue from the $K$-markets ERM strategy, despite its simplicity, Theorems \ref{thm:lower-bound-pd-1} and \ref{thm:lower-bound-pd} imply that the revenue from this strategy achieves the optimal rate of convergence (as a function of $n$) to the revenue from the true-distribution optimal 3PD strategy uniformly over $\mathcal{F}$. In other words, more sophisticated pricing strategies (e.g., with partitioning the covariate space based on observed frequencies) cannot improve upon the $K$-market ERM algorithm asymptotically.

To prove the lower bounds, we convert the problem into a classification task that tries to distinguish between distributions that are sufficiently close to each other but yield significantly different optimal prices. This technique was used in \cite{huang2018making}; there, the bound concerns data-based uniform pricing strategies, which only require constructing two distributions and simpler techniques. To establish the lower bound in Theorem \ref{thm:lower-bound-pd}, two distributions are far from being enough. The reason is that, unlike the uniform pricing problem where the optimal pricing function is a scalar parameter, the 3PD problem tries to learn an optimal pricing function of the covariate (an infinitely-dimensional parameter) and the deficiency in the \textit{expected} revenue concerns the entire pricing function at all covariate values. The notion of packing sets in \cite{kolmogorov1959varepsilon} and the Gilbert-Varshamov bound from coding theory are useful ingredients for proving Theorem \ref{thm:lower-bound-pd}. The most intricate part of the proof involves carefully constructing $M$ conditional densities (where $M$ grows with $n$) and bounding the separation between the optimal prices associated with these densities. The desired set of optimal prices in our proof is a packing set where the separation between elements is $\Omega(n^{-1/4})$ with respect to the unweighted $L_2$ norm, and the cardinality of this set is $\Omega(2^{n^{1/4}})$.

\subsection{Uniform pricing}

\label{sec:lower-bound-unif}


We have the following theorem for uniform pricing.

\begin{theorem} \label{thm:lower-bound-uniform} Let \cref{assumption:i.i.d.} hold. For any $\mathcal{F}^{U}$
satisfying the conditions in \cref{def:uniform} with $C^{*}\in(0,2)$
in (\ref{eq:2nd derivative}), the minimax difference in the revenues
is bounded from below as 
\begin{align*}
\mathscr{R}_{n}^{U}(\mathcal{F}^{U})\gtrsim n^{-2/3}.
\end{align*}

\end{theorem}

Theorem \ref{thm:lower-bound-uniform} states that there is an inevitable
deficiency, $\Omega(n^{-2/3})$, in the revenue
from \textit{any} data-based uniform pricing strategy relative to the revenue from
the true-distribution optimal uniform pricing strategy by taking
the supremum over $\mathcal{F}^{U}$.

Recalling Corollary          \ref{thm:upper-bound-unif} on the convergence rate $O(n^{-2/3})$ of the $1$-market ERM strategy, despite its simplicity, Theorem \ref{thm:lower-bound-uniform} implies that the revenue from this algorithm achieves the optimal rate of convergence (as a function of $n$) to the revenue from the true-distribution optimal uniform pricing strategy uniformly over $\mathcal{F}^{U}$. 

\subsection{Sketches of the proofs}
To facilitate understanding, we start with a preliminary of the proof for Theorem \ref{thm:lower-bound-pd-1} before laying out the preliminaries for Theorems \ref{thm:lower-bound-pd} and \ref{thm:lower-bound-uniform}.

\subsubsection{Preliminary of the proof for Theorem \ref{thm:lower-bound-pd-1}}

\label{sec:proof-sketch1}

For Theorem \ref{thm:lower-bound-pd-1}, we first show that the minimax
difference in price at a given covariate value $x_{0}$ is bounded
from below as follows:
\begin{align}
\inf_{\check{p}_{D} \in \check{\mathcal{D}}}\sup_{F_{Y,X}\in\mathcal{F}}\mathbb{E}_{F_{Y,X}}|\check{p}_{D}(x_{0};\textit{data})-p_{D}^{*}(x_{0};F_{Y,X})|\gtrsim n^{-1/4},x_{0}\in(0,1).\label{eq:diff_price}
\end{align}
Using Taylor expansion type of arguments and condition (\ref{eq:2nd derivative}),
we can relate the revenue difference to the minimax squared difference
in price at $x_{0}$: 
\begin{align*}
\mathscr{R}_{n}^{D}(x_{0};\mathcal{F}) & \gtrsim\inf_{\check{p}_{D} \in \check{\mathcal{D}}}\sup_{F_{Y,X}\in\mathcal{F}}\mathbb{E}_{F_{Y,X}}\left[|\check{p}_{D}(x_{0};\textit{data})-p_{D}^{*}(x_{0};F_{Y,X})|^{2}\right]\\
 & \geq\inf_{\check{p}_{D} \in \check{\mathcal{D}}}\sup_{F_{Y,X}\in\mathcal{F}}\left\{ \mathbb{E}_{F_{Y,X}}\left[|\check{p}_{D}(x_{0};\textit{data})-p_{D}^{*}(x_{0};F_{Y,X})|\right]\right\} ^{2}
\end{align*}
where the last line follows from the Jensen's inequality. 

The derivation of the lower bound (\ref{eq:diff_price}) can be reduced
to a binary classification problem. In a binary classification problem,
we have two distributions $F_{Y,X}^{1},F_{Y,X}^{2}\in\mathcal{F}$
whose optimal prices are separated by some number $2\varepsilon$;
that is, 
\begin{equation}
|p_{D}^{*}(x_{0};F_{Y,X}^{j'})-p_{D}^{*}(x_{0};F_{Y,X}^{j})|\geq2\varepsilon,\quad j,j'\in\{1,2\}.\label{eq:1}
\end{equation}
A binary classification rule uses the data to decide whether the true
distribution is $F_{Y,X}^{1}$ or $F_{Y,X}^{2}$. To relate the binary
classification problem to the pricing problem, note that, given any
pricing function $\check{p}_{D}$, we can use it to distinguish between
$F_{Y,X}^{1}$ and $F_{Y,X}^{2}$ in the following way. Define the
binary classification rule 
\begin{align*}
\psi(\textit{data})=\argmin_{j\in\{1,2\}}|p_{D}^{*}(x_{0};F_{Y,X}^{j})-\check{p}_{D}(x_{0};\textit{data})|.
\end{align*}
We claim that when the underlying distribution is $F_{Y,X}^{j}$,
the decision rule $\psi$ is correct if 
\begin{equation}
|p_{D}^{*}(x_{0};F_{Y,X}^{j})-\check{p}_{D}(x_{0};\textit{data})|<\varepsilon.\label{eq:2}
\end{equation}
To see this, note that by the triangle inequality, (\ref{eq:1}) and
(\ref{eq:2}) guarantee that 
\begin{align*}
 & |p_{D}^{*}(x_{0};F_{Y,X}^{j'})-\check{p}_{D}(x_{0};\textit{data})|\\
\geq & |p_{D}^{*}(x_{0};F_{Y,X}^{j'})-p_{D}^{*}(x_{0};F_{Y,X}^{j})|-|p_{D}^{*}(x_{0};F_{Y,X}^{j})-\check{p}_{D}(x_{0};\textit{data})|\\
> & 2\varepsilon-\varepsilon=\varepsilon,\text{ where }j'\ne j,\,j,j'\in\{1,2\}.
\end{align*}
This implies that 
\begin{align*}
\mathbb{P}_{F_{Y,X}^{j}}(\psi(\textit{data})\ne j)\leq\mathbb{P}_{F_{Y,X}^{j}}(|p_{D}^{*}(x_{0};F_{Y,X}^{j})-\check{p}_{D}(x_{0};\textit{data})|\geq\varepsilon),\quad j=1,2.
\end{align*}
Therefore, we can upper bound the average probability of mistakes
in the binary classification problem as 
\begin{align*}
 & \frac{1}{2}\mathbb{P}_{F_{Y,X}^{1}}(\psi(\textit{data})\ne1)+\frac{1}{2}\mathbb{P}_{F_{Y,X}^{2}}(\psi(\textit{data})\ne2)\\
\leq & \frac{1}{2}\mathbb{P}_{F_{Y,X}^{1}}(|p_{D}^{*}(x_{0};F_{Y,X}^{1})-\check{p}_{D}(x_{0};\textit{data})|\geq\varepsilon)+\frac{1}{2}\mathbb{P}_{F_{Y,X}^{2}}(|p_{D}^{*}(x_{0};F_{Y,X}^{2})-\check{p}_{D}(x_{0};\textit{data})|\geq\varepsilon)\\
\leq & \sup_{F_{Y,X}\in\mathcal{F}}\mathbb{P}_{F_{Y,X}}(|p_{D}^{*}(x_{0};F_{Y,X})-\check{p}_{D}(x_{0};\textit{data})|\geq\varepsilon).
\end{align*}
By the Markov inequality, we have 
\begin{align*}
 & \sup_{F_{Y,X}\in\mathcal{F}}\mathbb{E}|\check{p}_{D}(x_{0};\textit{data})-p_{D}^{*}(x_{0};F_{Y,X})|\\
\geq & \varepsilon\sup_{F_{Y,X}\in\mathcal{F}}\mathbb{P}\left(|\check{p}_{D}(x_{0};\textit{data})-p_{D}^{*}(x_{0};F_{Y,X})|\geq\varepsilon\right)\\
\geq & \varepsilon\left(\frac{1}{2}\mathbb{P}_{F_{Y,X}^{1}}(\psi(\textit{data})\ne1)+\frac{1}{2}\mathbb{P}_{F_{Y,X}^{2}}(\psi(\textit{data})\ne2)\right).
\end{align*}
Finally, we take the infimum over all pricing strategies on the left-hand
side (LHS), and the infimum over the induced set of binary decisions
on the right-hand side (RHS). This leads to 
\begin{align}
 & \inf_{\check{p}_{D} \in \check{\mathcal{D}}}\sup_{F_{Y,X}\in\mathcal{F}}\mathbb{E}|\check{p}_{D}(x_{0};\textit{data})-p_{D}^{*}(x_{0};F_{Y,X})|\nonumber \\
\geq & \varepsilon\inf_{\psi}\left(\frac{1}{2}\mathbb{P}_{F_{Y,X}^{1}}(\psi(\textit{data})\ne1)+\frac{1}{2}\mathbb{P}_{F_{Y,X}^{2}}(\psi(\textit{data})\ne2)\right).
\end{align}
The RHS of the above inequality consists of two parts: (1) $\varepsilon$, related to the separation
between two optimal prices, and (2) the average probability
of making a mistake in distinguishing the two distributions. To obtain
a meaningful bound, we want to find two distributions $F_{Y,X}^{1}$
and $F_{Y,X}^{2}$ that are close to each other (hard to distinguish)
but their optimal prices are sufficiently separated. We leave the details of the construction of such distributions
to the \textit{proof of Theorem \ref{thm:lower-bound-pd-1}} given
in Appendix \ref{sec:proof-lower-x0}.

\subsubsection{Preliminary of the proof for Theorem \ref{thm:lower-bound-pd}}

\label{sec:proof-sketch2}

For Theorem \ref{thm:lower-bound-pd}, we first show that the minimax
(unweighted) $L_{2}-$distance in price is bounded from below as follows:
\begin{align*}
\inf_{\check{p}_{D} \in \check{\mathcal{D}}}\sup_{F_{Y,X}\in\mathcal{F}}\mathbb{E}\lVert\check{p}_{D}(\textit{data})-p_{D}^{*}(F_{Y,X})\rVert_{2}^{2} & \gtrsim n^{-1/2}
\end{align*}
where 
\begin{align*}
\lVert\check{p}_{D}(\textit{data})-p_{D}^{*}(F_{Y,X})\rVert_{2}^{2}=\int_{0}^{1}|\check{p}_{D}(x;\textit{data})-p_{D}^{*}(x;F_{Y,X})|^{2}dx.
\end{align*}
Using Taylor expansion type of arguments and condition (\ref{eq:2nd derivative}),
we can relate the difference in the \textit{expected} revenues to
the minimax (unweighted) $L_{2}-$distance in price: 
\begin{align*}
\mathscr{R}_{n}^{D}(\mathcal{F})\gtrsim\inf_{\check{p}_{D} \in \check{\mathcal{D}}}\sup_{F_{Y,X}\in\mathcal{F}}\mathbb{E}_{F_{Y,X}}\lVert\check{p}_{D}(\textit{data})-p_{D}^{*}(F_{Y,X})\rVert_{2}^{2}
\end{align*}
where the expectation $\mathbb{E}_{F_{Y,X}}$ is taken with respect
to $\textit{data}\sim F_{Y,X}$. 

The object above concerns the entire pricing function $p_{D}^{*}(\cdot;F_{Y,X})$.
As a result, bounding the RHS of the above inequality is more complicated
than the previous one (\ref{eq:diff_price}). In particular, we consider
a multiple classification problem that tries to distinguish among
$M$ distributions, where $M$ is a function of the sample size $n$.
Similar as before, we want the optimal prices of these $M$ distributions
to be sufficiently separated. Similar derivations show that
the lower bound of the revenue problem can be reduced to that of a
multiple classification problem: 
\begin{align}
\inf_{\check{p}_{D} \in \check{\mathcal{D}}}\sup_{F_{Y,X}\in\mathcal{F}}\mathbb{E}_{F_{Y,X}}\lVert\check{p}_{D}(\textit{data})-p_{D}^{*}(F_{Y,X})\rVert_{2}^{2}\geq\varepsilon^{2}\inf_{\psi}\frac{1}{M}\sum_{j=1}^{M}\mathbb{P}_{F_{Y,X}^{j}}(\psi(\textit{data})\ne j),\label{eq:lower-bound-multiple-decision}
\end{align}
where the infimum $\inf_{\psi}$ is taken over the set of all multiple
decisions (with $M$ choices). To proceed, we apply the Fano's inequality
from information theory \citep{cover2005elements}. Fano's inequality
gives a lower bound on the average probability of mistakes:\footnote{We do not present the Fano's inequality in its standard form as in
\citet{cover2005elements}. Instead, we use a version from \citet{wainwright2019high}
that is more convenient for our purposes.} 
\begin{align}
\frac{1}{M}\sum_{j=1}^{M}\mathbb{P}_{F_{Y,X}^{j}}(\psi(\textit{data})\ne j)\geq1-\frac{\sum_{j,j'=1}^{M}\text{KL}(F_{Y,X}^{j}\|F_{Y,X}^{j'})/M^{2}+\log2}{\log M},\label{eq:fano}
\end{align}
where $\text{KL}(\cdot\|\cdot)$ denotes the Kullback-Leibler (KL)
divergence between two distributions: 
\begin{align*}
\text{KL}(F_{1}\|F_{2})\equiv\int f_{1}(y,x)\log\frac{f_{1}(y,x)}{f_{2}(y,x)}dydx.
\end{align*}
To obtain a sharp bound based on the multiple classification problem,
we want to find a set of distributions (where the cardinality $M$
of the set is large enough) that are close enough to each other (small enough pairwise
KL divergence) but their optimal prices are sufficiently separated. We leave the detailed proof to
Appendix \ref{sec:proof-lower-x0}. Our proof is based on a delicate
construction of conditional densities along with an application of
the Gilbert-Varshamov Lemma from coding theory. Specifically, we use the distribution $Y,X\sim U[0,1]$ with $X$ independent of $Y$ as the benchmark distribution and construct its perturbed versions with some correlation. 

\subsubsection{Preliminary of the proof for Theorem \ref{thm:lower-bound-uniform}}

Relative to the proofs in the case of 3PD, the proofs for the price-
and revenue-deficiency lower bounds in uniform pricing are simpler. We first
show that the minimax difference in price is bounded from below as
follows:
\begin{align}
\inf_{\check{p}_{U}\in\check{\mathcal{U}}}\sup_{F_{Y}\in\mathcal{F}^{U}}\mathbb{E}_{F_{Y}}|\check{p}_{U}(\textit{data}_{Y})-p_{U}^{*}|\gtrsim n^{-1/3}.\label{eq:diff_unif_price}
\end{align}
As previously, we can relate the revenue difference to the minimax
squared difference in price: 
\begin{align*}
\mathscr{R}_{n}^{U}(\mathcal{F}^{U}) & \gtrsim\inf_{\check{p}_{U}\in\check{\mathcal{U}}}\sup_{F_{Y}\in\mathcal{F}^{U}}\mathbb{E}_{F_{Y}}\left[|\check{p}_{U}(\textit{data}_{Y})-p_{U}^{*}|^{2}\right]\\
 & \geq\inf_{\check{p}_{U}\in\check{\mathcal{U}}}\sup_{F_{Y}\in\mathcal{F}^{U}}\left\{ \mathbb{E}_{F_{Y}}\left[|\check{p}_{U}(\textit{data}_{Y})-p_{U}^{*}|\right]\right\} ^{2}
\end{align*}
where the last line follows from the Jensen's inequlity. The derivation
of (\ref{eq:diff_unif_price}) only requires constructing two distributions,
similar to the approach discussed in Section \ref{sec:proof-sketch1}.

\section{Numerical evidence}\label{sec:numerical}

Sections \ref{sec:upper} and \ref{sec:lower} establish that the $K$-market ERM strategy achieves the optimal rates of convergence in revenue uniformly over a class of distributions. In this section, we turn to specific distributions and study the revenue performance of our $K$-markets ERM strategies in these cases. We present numerical evidence that supports the implications of our theoretical results.
Specifically, we calculate the revenues of the pricing strategies proposed in Section \ref{sec:upper} using real-world and simulated data.
We describe the data in detail below.
\paragraph{Data.}
For the empirical study, we use an eBay auction data set \citep{jank2010modeling}.
Because eBay uses a sealed-bid second-price auction format, the bid of each participant can serve as a proxy for an individual valuation of the object.
In particular, we use the data on 194 7-day auctions for the new Palm Pilot M515 PDAs.\footnote{\cite{jank2010modeling} also provide data on Cartier wristwatches, Swarovski beads, and Xbox game consoles, but each of these data sets may pool various configurations or models of these products categories. Thus, we choose the data on the Palm Pilot M515 to minimize such variations.}
The data has 3,832 observations at the bid level, and each observation includes an auction id, a bid amount, a bidder id, and a bidder rating.
Some bidders appear in the data set several times because either they revised their bid during an auction or participated in several auctions.
To be consistent with our assumption of independent sampling, we analyze the data at the bidder level and use the highest bid of each bidder across all auctions she participated in as the one representing her valuation.
This leaves 1,203 observations from which we draw samples of various sizes.
For $Y_i$, we use the bid (as described above) of bidder $i$ normalized to $[0,1]$.
For $X_i$ in the 3PD case, we use bidder $i$'s rating on eBay, which indicates the number of times sellers left feedback after a transaction with~$i$.

For the simulation study, we let the marginal distribution of $X$ be uniform on $[0,1]$ and the CDF of $Y$ conditional on $X=x$ be 
\begin{align}
  F_{Y|X}(y|x) &= y^{x+1}. \label{eq:simulated-cdf}
\end{align}

\paragraph{Implementation.} For each type of data, we calculate (a Monte-Carlo approximation of) the expected revenue generated by the uniform ERM and the $K$-markets ERM strategies for various sample sizes as follows.
First, fix $n$ and $K$.
Then, draw a sample $\{Y_i,X_i\}_{i=1}^n$ and, for each $k=1,\ldots, K$, let
\begin{align*}
  \text{market}_k \equiv \{Y_i\colon X_i\in I_k\}, \quad \hat F_k(t) \equiv \frac{|Y_i\in \text{market}_k\colon Y_i \leq t|}{|\text{market}_k|}.
\end{align*}

Then, the empirical optimal price in the $k$th market is given by
\begin{align*}
  \hat p_{D,k} 
  \equiv \argmax_{y\in [0,1]} y(1-\hat F_k(y)) 
  =\argmax_{y\in \text{market}_k} y(1-\hat F_k(y)),
\end{align*}
where the second equality holds because $\hat F_k$ is a step function.
Note that the uniform ERM strategy simply corresponds to the $1$-market ERM strategy. When $K>1$ and a drawn sample results in empty markets that contain no observations, we set the prices in those markets to one. 
Finally, we compute the revenue deficiency for the uniform ERM and $K$-markets ERM strategies (under $K \asymp n^{1/4}$).
\paragraph{Numerical findings.} Figure \ref{fig:revenue} plots the expected revenue generated by the $K$-markets ERM strategy for $K\in\{1,\ldots,5\}$ as a function of the sample size $n$ (with $K=1$ corresponding to the uniform ERM strategy).
To facilitate the exposition, we use a logarithmic scale for the $n$-axis.
For both types of data, one can see that for sufficiently small $n$, the $K$-markets revenue is \emph{decreasing} in $K$.
As $n$ grows, the performance of higher $K$ improves faster than that of lower $K$, and, for sufficiently large $n$, the $K$-markets revenue overtakes that with any lower $K$.
This finding can be explained by the bound $(K/n)^{2/3}+1/K^{2}$ in Theorem \ref{thm:upper-bound-pd}(ii), which implies that higher $K$ (more discrimination) approximates the revenue generated by the $F_{Y,X}$-optimum better but incurs a larger ``variance''. When the sample size is small, a lower $K$ can indeed be more beneficial. 

Figure \ref{fig:revenue} also suggests that, even if $X$ contains useful information about $Y$, the uniform ERM strategy may be revenue superior to any $K(>1)$-markets ERM strategy when $n$ is sufficiently small. Recall from Theorem \ref{thm:upper-bound-pd} that the bound $(K/n)^{2/3}+1/K^{2}$ is minimized at $K=n^{1/4}$, which gives $n^{-1/2}$, the optimal rate of convergence to the revenue generated by the $F_{Y,X}$-optimal 3PD strategy. This convergence rate is slower than $n^{-2/3}$, the optimal rate of convergence to the revenue generated by the $F_{Y}$-optimal uniform pricing strategy (c.f. Corollary \ref{thm:upper-bound-unif}). The slower convergence of the rate-optimal $K$-market ERM strategy can potentially dominate the revenue gain from price discrimination over without discrimination for small $n$.


Figure \ref{fig:regret} illustrates the difference in the convergence rates of the uniform ERM and the $K$-markets ERM strategies to their respective theoretical benchmarks.
In particular, we set $K = \frac15\lfloor n^{1/4} \rfloor$ for the simulation study and $K = \max\{1,\lfloor 2n^{1/4} - 7\rfloor\}$ for the empirical study.
As predicted by the rate $n^{-1/2}$ in Theorem \ref{thm:upper-bound-pd} and the rate $n^{-2/3}$ in Corollary \ref{thm:upper-bound-unif}, the revenue from the uniform ERM strategy is converging to the revenue from the $F_{Y}$-optimal uniform pricing strategy faster than the $K$-markets revenue to the revenue from the $F_{Y,X}$-optimal 3PD strategy.

Figure~\ref{fig:independent} exhibits the revenue under the $K$-markets ERM strategy for $K=1,\ldots,5$ and $n=2,\ldots, 10^5$, in the case where $X$ and $Y$ are uniform on $[0,1]$ and independent of each other. Not surprisingly, there is no benefit from price discrimination for revenue.

\def\optbpd{0.3542}
\def\optbunif{0.310}
\def\optspd{0.322992}
\def\optsunif{0.322343}

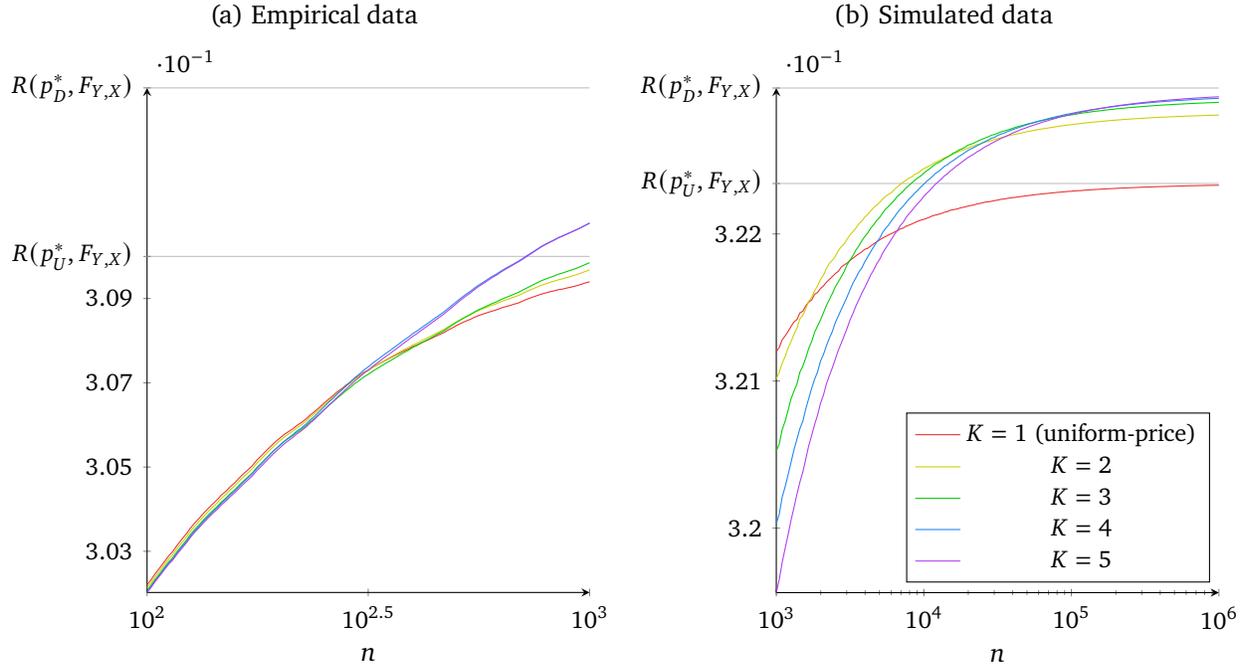
\begin{figure}[!htb]
\caption{Revenue under uniform and $K$-markets ERM strategies}\label{fig:revenue}
\begin{subfigure}{0.49\textwidth}
	\caption{Empirical data}\label{fig:revenue-empirical}
	\begin{tikzpicture}[scale=1]
      \begin{semilogxaxis}[
      		footnotesize,
      		xlabel={$n$},
          width=0.9\textwidth,
          height=\textwidth,
          axis x line=bottom,
          axis y line=left,
			xtick = {100, 316, 1000, 3162},
			ymax = \optbpd,
		    extra y ticks={\optbpd,\optbunif},
		    extra y tick style={grid=major},
		    extra y tick labels={{$R(p_{D}^{*},F_{Y,X})$},{$R(p_{U}^{*},F_{Y,X})$}},
			scaled x ticks=base 10:-2,
			scaled y ticks=base 10:1,
            /pgf/number format/precision=4,
          legend style={
              at={(0.98,0.02)},
              anchor=south east,legend columns=1
          },
      ]
      \addplot [mark=none,Firebrick2] file {figs/revenue_b_1.csv};
      \addplot [mark=none,Yellow3] file {figs/revenue_b_2.csv};
      \addplot [mark=none,Green3] file {figs/revenue_b_3.csv};
      \addplot [mark=none,DodgerBlue2] file {figs/revenue_b_4.csv};
      \addplot [mark=none,DarkOrchid2] file {figs/revenue_b_5.csv};
      \end{semilogxaxis}
  \end{tikzpicture}
\end{subfigure}
\begin{subfigure}{0.49\textwidth}
	\caption{Data simulated from \eqref{eq:simulated-cdf}}\label{fig:revenue-sim}
	\begin{tikzpicture}[scale=1]
      \begin{semilogxaxis}[
      		footnotesize,
      		xlabel={$n$},
          width=0.9\textwidth,
          height=\textwidth,
          axis x line=bottom,
          axis y line=left,
			ytick = {0.319, 0.32, 0.321, 0.322},
			ymax = \optspd,
		    extra y ticks={\optspd,\optsunif},
		    extra y tick style={grid=major},
		    extra y tick labels={{$R(p_{D}^{*},F_{Y,X})$},{$R(p_{U}^{*},F_{Y,X})$}},
			scaled x ticks=base 10:-1,
			scaled y ticks=base 10:1,
            /pgf/number format/precision=3,
          legend style={
              at={(0.98,0.02)},
              anchor=south east,legend columns=1
          },
      ]

      \addplot [mark=none,Firebrick2] file {figs/revenue_s_1.csv};\addlegendentry{\footnotesize $K = 1$ (uniform)}
      \addplot [mark=none,Yellow3] file {figs/revenue_s_2.csv};\addlegendentry{\footnotesize $K=2$}
      \addplot [mark=none,Green3] file {figs/revenue_s_3.csv};\addlegendentry{\footnotesize $K=3$}
      \addplot [mark=none,DodgerBlue2] file {figs/revenue_s_4.csv};\addlegendentry{\footnotesize $K=4$}
      \addplot [mark=none,DarkOrchid2] file {figs/revenue_s_5.csv};\addlegendentry{\footnotesize $K=5$}
      \end{semilogxaxis}
  \end{tikzpicture}
\end{subfigure}

\end{figure}

\begin{figure}[!htb]
\caption{Data-based revenue deficiency under uniform and $K$-markets ERM strategies (with $K\asymp n^{1/4}$).}\label{fig:regret}
\begin{subfigure}{0.49\textwidth}
	\caption{Empirical data}\label{fig:regret-empirical}
	\begin{tikzpicture}[scale=1]
      \begin{semilogxaxis}[
      		footnotesize,
      		xlabel={$n$},
          width=\textwidth,
          height=\textwidth,
          axis x line=bottom,
          axis y line=left,
			xtick = {100, 316, 1000},
			ytick = {0, 0.01, 0.02, 0.03, 0.04},
			ymin = 0,
			scaled x ticks=base 10:-2,
			scaled y ticks=base 10:2,
            /pgf/number format/precision=4,
          legend style={
              at={(0.99,0.99)},
              legend columns=5
          },
      ]
      \addplot [mark=none,black] file {figs/regret_b_uniform.csv};
       \addplot [mark=none,Firebrick2] file {figs/regret_b_1.csv};
      \addplot [mark=none,Yellow3] file {figs/regret_b_2.csv};
      \addplot [mark=none,Green3] file {figs/regret_b_3.csv};
      \addplot [mark=none,DodgerBlue2] file {figs/regret_b_4.csv};
      \addplot [mark=none,DarkOrchid2] file {figs/regret_b_5.csv};
      \end{semilogxaxis}
  \end{tikzpicture}
\end{subfigure}
\begin{subfigure}{0.49\textwidth}
	\caption{Data simulated from \eqref{eq:simulated-cdf}}\label{fig:regret-sim}
	\begin{tikzpicture}[scale=1]
     \begin{semilogxaxis}[
     		footnotesize,
     		xlabel={$n$},
         width=\textwidth,
         height=\textwidth,
         axis x line=bottom,
         axis y line=left,
         ytick = {0, 0.0004, 0.0008, 0.0012, 0.0016},
         ymin = 0,
          scaled ticks=true,
			scaled ticks=true,
			scaled x ticks=base 10:-1,
			scaled y ticks=base 10:3,
           /pgf/number format/precision=3,
         legend style={
             at={(0.98,0.98)},
             anchor=north east,
             legend columns=1
         },
     ]
     \addplot [mark=none,black,opacity=0.5] file {figs/regret_s_uniform.csv};\addlegendentry{\footnotesize uniform}
     \addplot [mark=dotted,dashed,black] file {figs/regret_s_uniform_fit.csv};\addlegendentry{\footnotesize $0.11\cdot n^{-2/3}$}
      \addplot [mark=none,Firebrick2] file {figs/regret_s_1.csv};\addlegendentry{\footnotesize $K=1$}
     \addplot [mark=none,Yellow3] file {figs/regret_s_2.csv};\addlegendentry{\footnotesize $K=2$}
     \addplot [mark=none,Green3] file {figs/regret_s_3.csv};\addlegendentry{\footnotesize $K=3$}
     \addplot [mark=none,DodgerBlue2] file {figs/regret_s_4.csv};\addlegendentry{\footnotesize $K=4$}
     \addplot [mark=none,DarkOrchid2] file {figs/regret_s_5.csv};\addlegendentry{\footnotesize $K=5$}
     \addplot [mark=none,black, dotted] file {figs/regret_s_fit.csv};\addlegendentry{\footnotesize $0.06 \cdot n^{-1/2}$}
     \end{semilogxaxis}
  \end{tikzpicture}
\end{subfigure}

\end{figure}
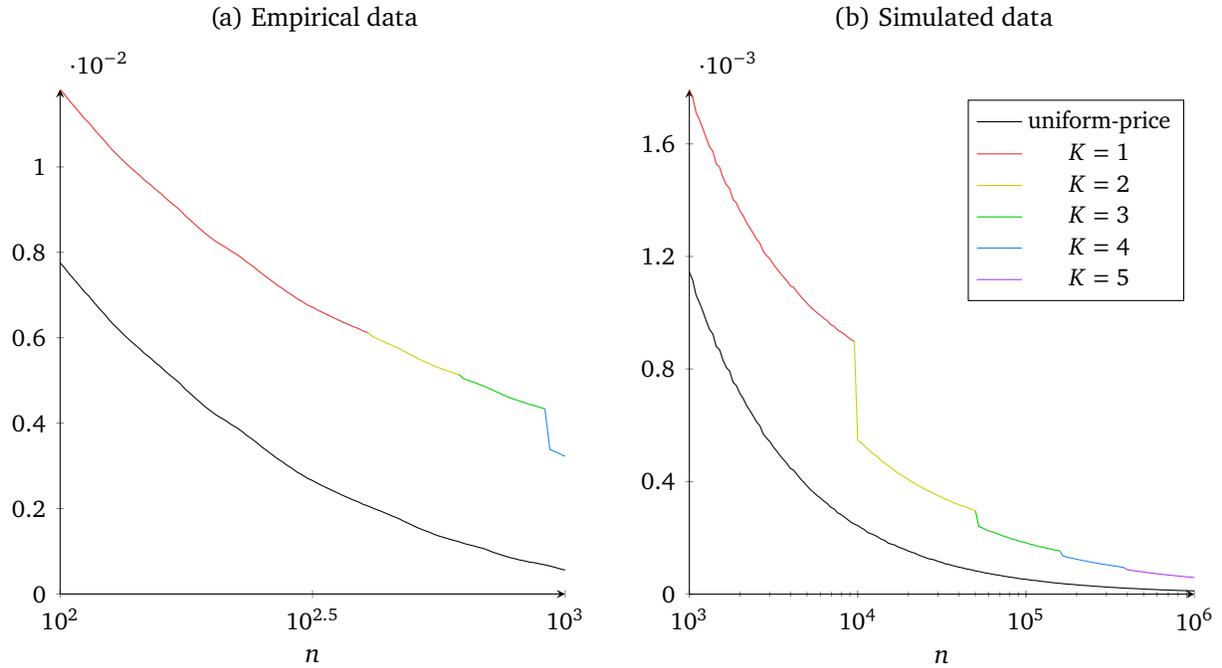

\begin{figure}[t]
	\caption{Uniform and $K$-markets revenue for the case of $X$ and $Y$ uniform on $[0,1]$ and independent of each other.}\label{fig:independent}
	\begin{center}
		\begin{tikzpicture}[scale=1]
	      \begin{semilogxaxis}[
	      		footnotesize,
	      		xlabel={$n$},
	          width=0.49\textwidth,
	          height=0.49\textwidth,
	          axis x line=bottom,
	          axis y line=left,
				xtick = {100, 1000, 10000, 10000},
				extra x ticks = {2, 10},
				extra x tick labels = {2, 10},
				ytick = {0,0.05,0.1,0.15,0.2,0.25},
				ymax = 0.25,
				ymin = 0,
			    extra y ticks={0.25},
			    extra y tick style={grid=major},
				scaled x ticks=base 10:-2,
	            /pgf/number format/precision=4,
	          legend style={
	              at={(0.98,0.02)},
	              anchor=south east,legend columns=1
	          },
	      ]
	      \addplot [mark=none,Firebrick2] file {figs/revenue_i_1.csv};\addlegendentry{$K = 1$ (uniform)}
	      \addplot [mark=none,Yellow3] file {figs/revenue_i_2.csv};\addlegendentry{$K=2$}
	      \addplot [mark=none,Green3] file {figs/revenue_i_3.csv};\addlegendentry{$K=3$}
	      \addplot [mark=none,DodgerBlue2] file {figs/revenue_i_4.csv};\addlegendentry{$K=4$}
	      \addplot [mark=none,DarkOrchid2] file {figs/revenue_i_5.csv};\addlegendentry{$K=5$}
	      \end{semilogxaxis}
	  \end{tikzpicture}
	\end{center}
\end{figure}


\section{Discussions}\label{sec:Discussions}

Recall that $p_{D}^{*}$ is the true-distribution optimal 3PD strategy and $\hat{p}_{D}$ is the $K$-markets ERM strategy with $K=\Theta(n^{1/4})$ giving the best trade-off between the ``variance'' and approximation error as shown in Theorem \ref{thm:upper-bound-pd}; $p_{U}^{*}$ is the true-distribution optimal uniform pricing strategy and $\hat{p}_{U}$ is the uniform ERM strategy. We can decompose the difference between the expected revenues generated respectively from $\hat{p}_{D}$ and $\hat{p}_{U}$ as follows:
\[
\mathbb{E}[R(\hat{p}_{D})]-\mathbb{E}[R(\hat{p}_{U})]=-\underbrace{\left\{ R(p_{D}^{*})-\mathbb{E}[R(\hat{p}_{D})]\right\} }_{A_{1}}+\underbrace{R(p_{D}^{*})-R(p_{U}^{*})}_{A_{2}}+\underbrace{R(p_{U}^{*})-\mathbb{E}[R(\hat{p}_{U})]}_{A_{3}}.	
\]
The first term $A_{1}=\Theta\left(n^{-1/2}\right)$ under a worst-case distribution $F_{Y,X}\in \mathcal{F}$, and the third term $A_{3}=O\left(n^{-2/3}\right)$ under $F_{Y}$, the marginal of $F_{Y,X}$. The second term $A_{2}=\Theta(1)$ when $X$ contains sufficient information about the valuation $Y$. Then, a sufficient condition for $\hat{p}_{D}$ to be revenue superior to $\hat{p}_{U}$ is that $n\rightarrow \infty$. In theory, this claim can be proved with the upper bounds in Section \ref{sec:upper} and a different construction in the derivations of the lower bounds. Particularly, this new construction would first find a density $f_{Y,X}$ such that the revenue generated by the corresponding $f_{Y,X}$-optimal 3PD strategy is well separated from the revenue generated by the optimal uniform pricing strategy associated with $f_{Y}$, and then build a large enough class of perturbed versions of $f_{Y,X}$; finally we would bound the separation between the optimal prices associated with these densities, in a similar fashion as what is done in Appendix \ref{sec: proofs for lower bounds}. In the paper, to make the analysis tractable, we choose the distribution $Y,X\sim U[0,1]$ with $X$ independent of $Y$ as the benchmark distribution and construct its perturbed versions with some correlation.

A challenging open question is, can the condition on $n$ be weakened to some finite number and if so, when? To answer this question, we would have to derive the universal constants in our bounds in meaningful forms. Unfortunately, due to the complexity of our problem, this exercise is infeasible under the
existing techniques from mathematical statistics, probability theory, and information theory. 


%

Our results suggest that it is more beneficial to engage in sample-based uniform pricing when $X$ is independent of $Y$.\footnote{The information of independence is unknown to the seller. He/she can statistically test for the independence of $Y$ and $X$ from the data but any such tests would suffer from Type I and Type II errors.} The fundamental reason lies in the proofs for Theorems \ref{thm:lower-bound-pd-1} and \ref{thm:lower-bound-pd}: unless $n=\infty $, no strategies that exploit $\{(Y_{i},X_{i}),1\leq i\leq n\}$ are able to distinguish with certainty the distribution $Y,X\sim U[0,1]$ with $X$ independent of $Y$ from its perturbed versions with some correlation (see the detailed constructions in Appendix \ref{sec: proofs for lower bounds}). The curse of dimensionality from exploiting the covariate $X$ makes the convergence of 3PD strategies based on $\{(Y_{i},X_{i}),1\leq i\leq n\}$ slower than that of the uniform pricing strategies based on $\{Y_{i},1\leq i\leq n\}$. 

Our upper
and lower bounds together suggest the following possibility: even when the covariate $X$ contains useful information about the valuation
$Y$, the $K$-markets ERM strategy \textit{can} be revenue inferior to the uniform ERM strategy in finite samples, due to the curse
of dimensionality and slower convergence of the $K$-markets ERM strategy
to its true-distribution optimal counterpart (and hence, a more stringent growth requirement of the sample size).
Indeed, the numerical evidence in Section \ref{sec:numerical} confirms this possibility. But such an implication should be taken with caution in small samples.

\paragraph{Small sample complication.} Given the pattern observed in our numerical studies, one might conjecture the following: there exists some $\bar n > 1$ such that when $n < \bar n$, the uniform ERM is always revenue-superior to the $K$-markets ERM (with $K>1$). In what follows, we explain why this conjecture may not hold.

Specifically, in our language, \cite{babaioff2018two} construct a distribution $F_Y$ such that the uniform ERM revenue under $n=2$ is strictly smaller than the uniform ERM revenue under $n=1$.
This seemingly counterintuitive result highlights the difficulty of establishing general comparative results with very small sample size.
We now argue that this construction also sheds some light on the comparison of the revenue performance of the $K$-markets ERM strategy with $K=1$ vs $K=2$ in the case of $n=2$.

To make this connection, we take $X$ to be uniform on $[0,1]$ and independent of $Y$, and assume that in the case $K=2$, when one of the markets is empty, the price for this market is set at the same level as for the other market.
Then, if $K=2$ and both markets are non-empty, the revenue in each market equals the $1$-market ERM revenue under $n=1$.
Otherwise, if both observations are in the same market, then the revenue equals the $1$-market ERM revenue under $n=2$.
Therefore, the expected $2$-markets ERM revenue with $n=2$ is the average of the $1$-market ERM revenue under $n=1$ and $n=2$ and hence strictly higher than the $1$-market ERM revenue with $n=2$ for a distribution $F_Y$ exhibiting the property discussed in \cite{babaioff2018two}.

More formally, let $R_{K,n}$ denote the expected revenue of the $K$-markets ERM strategy with a sample of size $n$.
 	Then, 
	\begin{align*}
		R_{2,2}
			&= \mathbb{E}_{\data^{n=2} \sim \,F_{Y,X}} [R(\hat{p}_D(\data),F_{Y,X})]\\
			&= \frac12 \mathbb{E}_{\data^{n=2} \sim  F_{Y,X} | I_{1} = \varnothing \text{ or } I_{2} = \varnothing} [R(\hat{p}_D(\data),F_{Y,X})]
			\\ &\quad 
			+ \frac12 \mathbb{E}_{\data^{n=2} \sim  F_{Y,X} | I_{1} \neq \varnothing \text{ and } I_{2} \neq \varnothing} [R(\hat{p}_D(\data),F_{Y,X})] \\
			&= \frac12   R_{1,2}
			 + \frac12 R_{1,1}.
	\end{align*}
    Therefore, $R_{1,1} > R_{1,2}$ implies $R_{2,2} > R_{1,2}$.
    
    Finally, we add the caveat that the construction in \cite{babaioff2018two} is based on an atomless approximation of the censored equal-revenue distribution $F_{Y}(y) = 1-1/y, y\in[1,\infty)$ which has a discontinuous density.
    However, it is straightforward to verify that the same property holds for the equal-revenue distribution truncated at any $y > 4$, which has a Lipschitz continuous and differentiable density.
    Moreover, the equal revenue distribution truncated at $y$ and translated to the left by $t>1/y$ (so that the support is $[1-t,y-t]$) also has a Lipschitz continuous and differentiable density, the interior optimal price (in line with our assumptions), and satisfies the \cite{babaioff2018two} property, e.g., for $y=4, t=1/2$.

    \paragraph{An open problem.} To conclude, we would like to propose a challenging open problem: Do there exist some $3 \leq \underline{n}<\bar{n}<\infty$ such that for any $n\in [\underline{n},\bar{n}]$ and distribution in $\mathcal{F}$, the $K$-markets ERM strategy (for any $K>1$) is always revenue-inferior to the uniform ERM strategy?

\newpage
\appendix
\numberwithin{lemma}{section}
\numberwithin{figure}{section}
\section{Proofs for upper bounds}

\label{sec:proof-upper-uniform}

To facilitate the presentation, we first give the proof for Corollary
\ref{thm:upper-bound-unif}. 
\begin{proof}[Proof of Corollary \ref{thm:upper-bound-unif}]

Denote $\kappa'\equiv\inf_{p\in[0,1]}|R''(p)|/2>0$. By
Taylor expansion, for any $p$, 
\begin{align*}
R(p_{U}^{*})-R(p)\geq\kappa'(p-p_{U}^{*})^{2}.
\end{align*}
Denote $\hat{R}(p)\equiv p(1-\hat{F}(p))$. Combining the inequality
above with the basic inequality (i.e., $\hat{R}(\hat{p}_{U})\geq\hat{R}(p_{U}^{*})$),
we have 
\begin{align}
\kappa'(\hat{p}_{U}-p_{U}^{*})^{2}\leq R(p_{U}^{*})-R(\hat{p}_{U})\leq R(p_{U}^{*})-\hat{R}(p_{U}^{*})-(R(\hat{p}_{U})-\hat{R}(\hat{p}_{U})).\label{eq:basic-ineq}
\end{align}
For $\delta\in(0,p_{U}^{*}]$, define 
\begin{align*}
\mathcal{G}_{\delta}\equiv\{y\mapsto p\mathbf{1}\{y\geq p\}-p_{U}^{*}\mathbf{1}\{y\geq p^{*}\}\colon p\in[p_{U}^{*}-\delta,p_{U}^{*}+\delta]\}
\end{align*}
and 
\begin{align*}
G_{\delta}(y)\equiv\begin{cases}
0, & \text{ if }y<p_{U}^{*}-\delta,\\
p_{U}^{*}, & \text{ if }p_{U}^{*}-\delta\leq y\leq p_{U}^{*}+\delta,\\
\delta, & \text{ if }y>p_{U}^{*}+\delta.
\end{cases}
\end{align*}
Then $G_{\delta}$ is an envelope function of the class $\mathcal{G}_{\delta}$.
The $L_{2}\left(P\right)-$norm of $G_{\delta}$ is bounded by 
\begin{align*}
\lVert G_{\delta}\rVert_{L_{2}\left(P\right)}=\left((p_{U}^{*})^{2}\mathbb{P}(Y\in[p_{U}^{*}-\delta,p_{U}^{*}+\delta])+\delta^{2}\mathbb{P}(Y>p_{U}^{*}+\delta)\right)^{1/2}\leq C\sqrt{\delta}.
\end{align*}
As we argue in the proof of Lemma \ref{lm:G_delta-VC-subgraph}, $\mathcal{G}_{\delta}$ is a VC-subgraph class, so we have 
\begin{align}
\mathbb{E}\sup_{g\in\mathcal{G}_{\delta}}\left|\frac{1}{n}\sum_{i=1}^{n}g(Y_{i})-\mathbb{E}g(Y_{i})\right|\leq C\sqrt{\delta/n}.\label{eq:expected-sup-EP}
\end{align}
We derive the convergence rate of $\hat{p}-p^{*}$ via a peeling argument.
Consider the following decomposition 
\begin{align*}
\mathbb{P}\left(n^{1/3}|\hat{p}_{U}-p_{U}^{*}|>M\right)=\sum_{j=M+1}^{\infty}\mathbb{P}\left(n^{1/3}|\hat{p}_{U}-p_{U}^{*}|\in(j-1,j]\right).
\end{align*}
For any $j\geq M+1$, we have 
\begin{align*}
 & \{|\hat{p}_{U}-p_{U}^{*}|\in((j-1)n^{-1/3},jn^{-1/3}]\}\\
= & \{|\hat{p}_{U}-p_{U}^{*}|>(j-1)n^{-1/3},|\hat{p}_{U}-p_{U}^{*}|\leq jn^{-1/3}\}\\
\subset & \left\{ R(p_{U}^{*})-\hat{R}(p_{U}^{*})-(R(\hat{p}_{U})-\hat{R}(\hat{p}_{U}))\geq\kappa'(j-1)^{2}n^{-2/3},|\hat{p}_{U}-p_{U}^{*}|\leq jn^{-1/3}\right\} \\
\subset & \left\{ \Delta_{j,n}\geq\kappa'(j-1)^{2}n^{-2/3}\right\} ,
\end{align*}
where the third line follows from (\ref{eq:basic-ineq}), and $\Delta_{j,n}$
in the last line is defined as 
\[
\Delta_{j,n}\equiv\sup_{g\in\mathcal{G}_{jn^{-1/3}}}\left|\frac{1}{n}\sum_{i=1}^{n}g(Y_{i})-\mathbb{E}g(Y_{i})\right|.
\]
Therefore, 
\begin{align*}
\mathbb{P}\left(|\hat{p}_{U}-p_{U}^{*}|\in((j-1)n^{-1/3},jn^{-1/3}]\right) & \leq\mathbb{P}\left(\Delta_{j,n}\geq\kappa'(j-1)^{2}n^{-2/3}\right).
\end{align*}
To bound the probability on the RHS of the above inequality, we use
the concentration inequality given by Theorem 7.3 in \citet{bousquet2003concentration},
which is a version of \citeauthor{talagrand1996new}'s (\citeyear{talagrand1996new})
inequality. The concentration inequality states that for all $t>0$,
\begin{align*}
\mathbb{P}\left(\Delta_{j,n}\geq\mathbb{E}\Delta_{j,n}+\sqrt{2t(\sigma^{2}+2\mathbb{E}\Delta_{j,n})/n}+t/(3n)\right)\leq\exp(-ct),
\end{align*}
for some universal constant $c>0$, where 
\begin{align*}
\sigma^{2}\equiv\sup_{g\in\mathcal{G}_{jn^{-1/3}}}\mathbb{E}g(Y_{1})^{2}\leq\lVert\mathcal{G}_{jn^{-1/3}}\rVert_{L_{2}}^{2}\leq Cjn^{-1/3}.
\end{align*}
From (\ref{eq:expected-sup-EP}), we have 
\begin{align*}
\mathbb{E}\Delta_{j,n}\leq C\sqrt{jn^{-1/3}/n}=C\sqrt{j}n^{-2/3}.
\end{align*}
By setting $t=\kappa'j^{2}$, we have 
\begin{align*}
 & \mathbb{E}\Delta_{j,n}+\sqrt{2t(\sigma^{2}+2\mathbb{E}\Delta_{j,n})/n}+t/(3n)\\
\leq & C\sqrt{j}n^{-2/3}+\sqrt{2\kappa'j^{2}(Cjn^{-1/3}+2C\sqrt{j}n^{-2/3})/n}+\kappa'j^{2}/(3n)\\
\leq & C'j^{3/2}n^{-2/3}\leq C^{*}(j-1)^{2}n^{-2/3},
\end{align*}
when $j$ is large enough. Then we have 
\begin{align*}
\mathbb{P}\left(\Delta_{j,n}\geq C^{*}(j-1)^{2}n^{-2/3}\right)\leq\mathbb{P}\left(\Delta_{j,n}\geq Cjn^{-2/3}\right)\leq\exp(-c\kappa'j^{2}),\text{ for \ensuremath{j} large}.
\end{align*}
To summarize, we have shown that 
\begin{align*}
\mathbb{P}\left(n^{1/3}|\hat{p}_{U}-p_{U}^{*}|>M\right)\leq\sum_{j=M+1}^{\infty}\exp(-C_{1}j^{2})\leq C_{3}\exp(-C_{2}M^{2}). 
\end{align*}
By integrating the tail probability, we have 
\begin{align}
\mathbb{E}|\hat{p}_{U}-p_{U}^{*}|^{s}\lesssim n^{-s/3}. \label{eq:price-diff-unif}
\end{align}
For revenue, we use the second-order Taylor expansion and obtain
that 
\begin{align*}
\mathbb{E}[R(p_{U}^{*})-R(\hat{p}_{U})]\leq\sup_{p}|R''(p)|\mathbb{E}(\hat{p}_{U}-p_{U}^{*})^{2}\lesssim n^{-2/3}.
\end{align*}
\end{proof}
\label{sec:proof-upper-3PD}
\begin{proof}[Proof of Theorem \ref{thm:upper-bound-pd}]

We introduce some notations. Let $\tilde{R}_{k}(p)$ denote the revenue
collected from the $k$th market by charging price $p$; that is,
\begin{align*}
\tilde{R}_{k}(p) & \equiv p\mathbb{P}(Y>p,X\in I_{k})\\
 & =p\int_{p}^{1}\int_{I_{k}}f_{Y|X}(y|x)f_{X}(x)dxdy.
\end{align*}
Denote $\tilde{p}_{k}\equiv\argmax_{p\in[0,1]}\tilde{R}_{k}(p)$ as
the maximizer of $\tilde{R}_{k}$. The first- and second-order
derivatives of $\tilde{R}_{k}$ are respectively 
\begin{align*}
\tilde{R}_{k}'(p) & =\int_{p}^{1}\int_{I_{k}}f_{Y|X}(y|x)f_{X}(x)dxdy-p\int_{I_{k}}f_{Y|X}(p|x)f_{X}(x)dx,\\
\tilde{R}_{k}''(p) & =\int_{I_{k}}\left(-2f_{Y|X}(p|x)-p\frac{\partial}{\partial y}f_{Y|X}(p|x)\right)f_{X}(x)dx.
\end{align*}
By the Lipschitz continuity assumption, the second-order derivative
$\tilde{R}_{k}''(p)$ exists for almost all $p\in[0,1]$. Recall
that 
\begin{align*}
-2f_{Y|X}(p|x)-p\frac{\partial}{\partial y}f_{Y|X}(p|x)\leq-C^{*},
\end{align*}
and $f_{X}$ is bounded away from zero. Denote $2\kappa''\equiv C^{*}\inf_{x\in[0,1]}f_{X}(x).$
Then 
\begin{align*}
\tilde{R}_{k}''(p)\leq-2\kappa''\int_{I_{k}}dx=-2\kappa''/K
\end{align*}
for almost all $p\in[0,1]$. By Lemma \ref{lm:foundamental-thm-calculus-Lipschitz},
we have 
\begin{align}
\tilde{R}_{k}(\tilde{p}_{k})-\tilde{R}_{k}(p)=|\tilde{R}_{k}(\tilde{p}_{k})-\tilde{R}_{k}(p)|\geq\frac{\kappa''}{K}(\tilde{p}_{k}-p)^{2},p\in[0,1].\label{eq:upper-pd-quadratic-lower}
\end{align}
Note that $\tilde{p}_{k}$ is not the true optimal price under $F_{Y,X}$.
We need to relate it to the true optimal price. Let $k(x_{0})$ be
such that $x_{0}\in I_{k}$. Then by the triangle inequality, we can
decompose the pricing difference into estimation error and approximation
error: 
\begin{align}
\mathbb{E}|\hat{p}_{D}(x_{0};\textit{data})-p_{D}^{*}(x_{0})| & =\mathbb{E}|\hat{p}_{k(x_{0})}-p_{D}^{*}(x_{0})|\nonumber \\
 & \leq\underset{\textrm{Estimation error}}{\underbrace{\mathbb{E}|\hat{p}_{k(x_{0})}-\tilde{p}_{k(x_{0})}|}}+\underset{\textrm{Approximation error}}{\underbrace{|\tilde{p}_{k(x_{0})}-p_{D}^{*}(x_{0})|}}.\label{eq:15}
\end{align}

\textbf{\textit{Estimation error}}. Denote $\hat{R}_{k}$ as the empirical
counterpart of $\tilde{R}_{k}$; that is, 
\begin{align*}
\hat{R}_{k}(p)\equiv\frac{p}{n_{k}}\sum_{i\in\{j:X_{j}\in I_{k}\}}\mathbf{1}\{Y_{i}>p,X_{i}\in I_{k}\}.
\end{align*}
Recall that $\hat{p}_{k}$ is the maximizer of $\hat{R}_{k}$. The
basic inequality (i.e., $\hat{R}_{k}(\hat{p}_{k})\geq\hat{R}_{k}(\tilde{p}_{k})$)
gives that 
\begin{align}
\tilde{R}_{k}(\tilde{p}_{k})-\tilde{R}_{k}(\hat{p}) & =\tilde{R}_{k}(\tilde{p}_{k})-\tilde{R}_{k}(\hat{p})-\hat{R}_{k}(\tilde{p}_{k})+\hat{R}_{k}(\tilde{p}_{k})\nonumber \\
 & \leq\tilde{R}_{k}(\tilde{p}_{k})-\tilde{R}_{k}(\hat{p})-\hat{R}_{k}(\tilde{p}_{k})+\hat{R}_{k}(\hat{p}_{k}).\label{eq:basic}
\end{align}
Combining (\ref{eq:upper-pd-quadratic-lower}) and (\ref{eq:basic})
yields 
\begin{align}
\frac{\kappa''}{K}(\tilde{p}_{k}-\hat{p}_{k})^{2}\leq\tilde{R}_{k}(\tilde{p}_{k})-\hat{R}_{k}(\tilde{p}_{k})-(\tilde{R}_{k}(\hat{p})-\hat{R}_{k}(\hat{p}_{k})).\label{eq:PD-basic-ineq}
\end{align}
In each $I_{k}$ ($k=1,\ldots,K$), the optimal price is the same.

In what follows, $s=1$ or $s=2$. Conditioning on
$X_{i}$ where $i$ falls in the $k$th market, the proof for Corollary
\ref{thm:upper-bound-unif}, in particular, (\ref{eq:price-diff-unif})
yields
\begin{equation}
\mathbb{E}\left[|\hat{p}_{k}-\tilde{p}_{k}|^{s}\vert X_{i},\,i\in\{j:X_{j}\in I_{k}\}\right]\lesssim(1/n_{k})^{s/3}.\label{eq:cond_price_error}
\end{equation}
By Assumption \ref{def:pd}(vi), the $i$th observation falls into
the $k$th market with probabilty $\frac{c_{1}}{K}$ and other markets
with probability $\frac{K-c_{1}}{K}$. Let us consider the event $\mathcal{A}_{q}=\left\{ n_{k}>qn\right\} $
where $q\in(0,\frac{c_{1}}{K})$. If $K$ is large enough such that
$\frac{c_{1}}{K}\leq\frac{1}{2}$, the classic binomial tail bound
yields 
\begin{equation}
\mathbb{P}(\mathcal{A}_{q})>1-\exp(-n\textrm{KL}(q||c_{1}/K)) \label{eq: A_q event}
\end{equation}
where 
\[
\textrm{KL}(q||c_{1}/K):=q\log\frac{qK}{c_{1}}+(1-q)\log\frac{(1-q)K}{K-c_{1}}.
\]
Therefore, we have 
\[
\mathbb{E}\left[|\hat{p}_{k}-\tilde{p}_{k}|^{s}\vert X_{i},\,i\in\{j:X_{j}\in I_{k}\}\right]\lesssim(qn)^{-s/3},\textrm{ for a given }k,
\]
with probability at least $1-\exp(-n\textrm{KL}(q||c_{1}/K))$. With
a union bound, we also have 
\begin{equation}
\mathbb{E}\left[|\hat{p}_{k}-\tilde{p}_{k}|^{s}\vert X_{i},\,i\in\{j:X_{j}\in I_{k}\}\right]\lesssim(qn)^{-s/3},\textrm{ for all }k,\label{eq:binomial}
\end{equation}
with probability at least $1-\exp(-n\textrm{KL}(q||c_{1}/K)+\log K)$.

Furthermore, we have 
\begin{equation}
\textrm{KL}(q||c_{1}/K)\geq\frac{1}{2}\left(\frac{c_{1}}{K}-q\right)^{2}.\label{eq:bin-KL}
\end{equation}
We show a more general result 
\[
\textrm{KL}(q||\alpha)=:g_{q}(\alpha)\geq\frac{\left(\alpha-q\right)^{2}}{2}
\]
for any $q\in(0,\alpha)$. Because $g_{q}(\cdot)$ is twice differentiable
and $g_{q}(q)=0$, a second-order Taylor expansion gives 
\[
g_{q}(\alpha)=g_{q}^{'}(q)(\alpha-q)+\frac{g_{q}^{''}(t)}{2}\left(\alpha-q\right)^{2}
\]
where $t\in[q,\alpha]$ and $g_{q}^{'}(t)=-\frac{q}{t}+\frac{1-q}{1-t}$.
Note that $g_{q}^{'}(q)=0$. Moreover, given $t\in(0,1)$ such that
$\frac{1}{t^{2}}\geq1$ and $\frac{1}{(1-t)^{2}}\geq1$, we have 
\[
g_{q}^{''}(t)=\frac{q}{t^{2}}+\frac{1-q}{(1-t)^{2}}\geq1.
\]

As a consequence of (\ref{eq:binomial}) and (\ref{eq:bin-KL}),
we have 
\begin{equation}
\mathbb{E}\left[|\hat{p}_{k}-\tilde{p}_{k}|^{s}\vert X_{i},\,i\in\{j:X_{j}\in I_{k}\}\right]\lesssim(qn)^{-s/3},\textrm{ for all }k,\label{eq:binomial-1}
\end{equation}
with probability at least $1-\exp\left(-n\left(\frac{c_{1}}{K}-q\right)^{2}/2+\log K\right)$.
Taking $q=\frac{c_{1}}{2K}$, (\ref{eq:binomial-1}) gives 
\begin{equation}
\mathbb{E}\left[|\hat{p}_{k}-\tilde{p}_{k}|^{s}\vert X_{i},\,i\in\{j:X_{j}\in I_{k}\}\right]\lesssim(K/n)^{s/3},\textrm{ for all }k,\label{eq:binomial-main}
\end{equation}
with probability at least $1-\exp\left(-\frac{nc_{1}^{2}}{8K^{2}}+\log K\right)$.

In view of (\ref{eq:cond_price_error}), (\ref{eq:binomial}),
(\ref{eq:binomial-1}) and (\ref{eq:binomial-main}), the source of
uncertainty from the conditioning is solely from the statistics $n_{k}$.
Using this fact, (\ref{eq:binomial-main}) as well as the fact that
$\hat{p}_{k}$ and $\tilde{p}_{k}$ are bounded, we have
\begin{eqnarray*}
\mathbb{E}\left[|\hat{p}_{k}-\tilde{p}_{k}|^{s}\right] & = & \mathbb{E}\left[|\hat{p}_{k}-\tilde{p}_{k}|^{s}1\left\{ n_{k}>\frac{nc_{1}}{2K}\right\} \right]+\mathbb{E}\left[|\hat{p}_{k}-\tilde{p}_{k}|^{s}1\left\{ n_{k}\leq\frac{nc_{1}}{2K}\right\} \right]\\
 & \precsim & (K/n)^{s/3}+\exp\left(-\frac{nc_{1}^{2}}{8K^{2}}+\log K\right).\label{eq:binomial_final}
\end{eqnarray*}

\textbf{\textit{Approximation error}}. The second term $|\tilde{p}_{k(x_{0})}-p_{D}^{*}(x_{0})|$
in (\ref{eq:15}) is deterministic and can be controlled by using
the smoothness conditions. By definition, $p_{D}^{*}(x_{0})$ satisfies
the first-order condition 
\begin{align*}
0=\frac{\partial}{\partial p}r(p_{D}^{*}(x),x).
\end{align*}
By the differentiability condition of $\mathcal{F}$, $\frac{\partial}{\partial p}r(p,x)$
is continuously differentiable in $(p,x)$ in a neighborhood of $(p_{D}^{*}(x),x)$.
By the strong concavity, $\frac{\partial^{2}}{\partial p^{2}}r(p_{D}^{*}(x),x)$
is non-zero. Then by the implicit function theorem, the function
$p_{D}^{*}(x)$ is well-defined (uniquely determined by the first-order
condition) and is differentiable. Its derivative is given as follows:
\begin{align*}
\frac{d}{dx}p_{D}^{*}(x)=-\frac{\frac{\partial^{2}}{\partial p\partial x}r(p_{D}^{*}(x),x)}{\frac{\partial^{2}}{\partial p^{2}}r(p_{D}^{*}(x),x)}.
\end{align*}
By the strong concavity, the absolute value of $\frac{\partial^{2}}{\partial p^{2}}r(p,x)$
is bounded away from zero; also, the function $\left|\frac{\partial}{\partial x}F_{Y|X}(y|x)+y\frac{\partial}{\partial x}f_{Y|X}(y|x)\right|$
is bounded above by $\bar{C}$. This implies that $p_{D}^{*}(x)$ is Lipschitz continuous
on $[0,1]$. We use $L_{1}$ to denote the Lipschitz constant. By
applying Taylor expansion to the first-order condition of $\tilde{p}_{k}$,
we have 
\begin{align*}
0=\int_{I_{k}}\frac{\partial}{\partial p}r(\tilde{p}_{k},x)f_{X}(x)dx=\int_{I_{k}}\underbrace{\frac{\partial}{\partial p}r(p_{D}^{*}(x),x)}_{=0}f_{X}(x)dx\\
+\int_{I_{k}}\frac{\partial^{2}}{\partial p^{2}}r(\bar{p}(x),x)(\tilde{p}_{k}-p_{D}^{*}(x))f_{X}(x)dx,
\end{align*}
for some $\bar{p}(x)$ between $\tilde{p}_{k}$ and $p_{D}^{*}(x)$.
Rearranging terms shows that $\tilde{p}_{k}$ is a weighted average
of $p_{D}^{*}(x),x\in I_{k}$; that is, 
\[
\tilde{p}_{k}=\frac{\int_{I_{k}}\frac{\partial^{2}}{\partial p^{2}}r(\bar{p}(x),x)p_{D}^{*}(x)f_{X}(x)dx}{\int_{I_{k}}\frac{\partial^{2}}{\partial p^{2}}r(\bar{p}(x),x)f_{X}(x)dx}.
\]
Since $p_{D}^{*}(x)$ is Lipschitz continuous, the triangle inequality implies that
\begin{align*}
|\tilde{p}_{k}-p_{D}^{*}(x_{0})|^{s}\leq L_{1}^{s}/K^{s},\text{ for any fixed }s\geq1.
\end{align*}

Therefore, we obtain the following upper bound 
\begin{align*}
\mathbb{E}|\hat{p}_{D}(x_{0};\textit{data})-p_{D}^{*}(x_{0})|^{2}\lesssim{1/K^{2}+(K/n)^{2/3}+\exp\left(-\frac{nc_{1}^{2}}{8K^{2}}+\log K\right).}
\end{align*}
By choosing $K\asymp n^{-1/4}$, the above bound becomes $n^{-1/4}$.

In addition, Lemma \ref{lm:foundamental-thm-calculus-Lipschitz}(iii) gives
that 
\begin{align*}
\mathbb{E}\left[r(p_{D}^{*},x_{0})-r(\hat{p}_{D}(\textit{data}),x_{0})\right] & \lesssim\mathbb{E}\left[|\hat{p}_{D}(x_{0};\textit{data})-p_{D}^{*}(x_{0})|^{2}\right]\\
 & \lesssim1/K^{2}+(K/n)^{2/3}+\exp\left(-\frac{nc_{1}^{2}}{8K^{2}}+\log K\right).
\end{align*}
This proves part (i) of the theorem.

For part (ii), we want to bound the expected revenue difference. Consider
the following decomposition: 
\begin{align*}
 & R(p_{D}^{*},F_{Y,X})-R(\hat{p}_{D}(\textit{data}),F_{Y,X})\\
\leq & R(p_{D}^{*},F_{Y,X})-R(\tilde{p},F_{Y,X})+|R(\tilde{p},F_{Y,X})-R(\hat{p}_{D}(\textit{data}),F_{Y,X})|.
\end{align*}
The first term on the RHS is deterministic and can be bounded by using
Lemma \ref{lm:foundamental-thm-calculus-Lipschitz}(iii) as follows:
\begin{align*}
 & |R(p_{D}^{*},F_{Y,X})-R(\tilde{p},F_{Y,X})|\\
\leq & \int_{0}^{1}|r(p_{D}^{*}(x),x)-r(\tilde{p}(x),x)|f_{X}(x)dx\\
= & \sum_{k=1}^{K}\int_{I_{k}}|r(p_{D}^{*}(x),x)-r(\tilde{p}_{k},x)|f_{X}(x)dx\\
\leq & \sum_{k=1}^{K}\int_{I_{k}}\frac{1}{2}\left|2f_{Y|X}(y|x)+y\frac{\partial}{\partial y}f_{Y|X}(y|x)\right|\sup_{y,x}|p_{D}^{*}(x)-\tilde{p}_{k}|^{2}f_{X}(x)dx\\
 & \lesssim1/K^{2}.
\end{align*}
where we have used the first-order condition of $p_{D}^{*}$. For
the second term, we have 
\begin{align*}
R(\tilde{p},F_{Y,X})-R(\hat{p}_{D}(\textit{data}),F_{Y,X})=\sum_{k=1}^{K}\tilde{R}_{k}(\tilde{p}_{k})-\tilde{R}_{k}(\hat{p}_{k}).
\end{align*}
This is because both $\hat{p}(\textit{data})$ and $\tilde{p}$ are
constant within each $I_{k}$. Their revenues on $I_{k}$ are reduced
to $\tilde{R}_{k}$. Note that for every $k$, $\tilde{R}_{k}'(\tilde{p}_{k})=0$,
and 
\begin{align*}
|\tilde{R}_{k}''(p)| & \leq\int_{I_{k}}\left|2f_{Y|X}(y|x)+y\frac{\partial}{\partial y}f_{Y|X}(y|x)\right|f_{X}(x)dx\\
 & \leq\frac{1}{K}\sup_{y,x}\left(\left|2f_{Y|X}(y|x)+y\frac{\partial}{\partial y}f_{Y|X}(y|x)\right|f_{X}(x)\right).
\end{align*}
Then Lemma \ref{lm:foundamental-thm-calculus-Lipschitz}(iii) gives
that 
\begin{align*}
\tilde{R}_{k}(\tilde{p}_{k})-\tilde{R}_{k}(\hat{p}_{k})\lesssim1/K(\tilde{p}_{k}-\hat{p}_{k})^{2}.
\end{align*}
Hence, we have 
\begin{align*}
\mathbb{E}|R(\tilde{p},F_{Y,X})-R(\hat{p}_{D}(\textit{data}),F_{Y,X})| & \leq\sum_{k=1}^{K}\mathbb{E}|\tilde{R}_{k}(\tilde{p}_{k})-\tilde{R}_{k}(\hat{p}_{k})|\\
 & \lesssim{(K/n)^{2/3}+\exp\left(-\frac{nc_{1}^{2}}{8K^{2}}+\log K\right).}
\end{align*}
To summarize, we have shown that 
\begin{align*}
R(p_{D}^{*},F_{Y,X})-R(\hat{p}_{D}(\textit{data}),F_{Y,X})\lesssim{1/K^{2}+(K/n)^{2/3}+\exp\left(-\frac{nc_{1}^{2}}{8K^{2}}+\log K\right).}
\end{align*}
By choosing $K\asymp n^{-1/4}$, the above bound becomes $n^{-1/2}$.
This proves part (ii) of the theorem.  
\end{proof}
\label{sec:proof-upper-welfare}
\begin{proof}[Proof of Theorem \ref{thm:welfare}]
For part (i), notice that the welfare can be written as a double
integral 
\begin{align*}
W(p,F_{Y,X})=\int_{0}^{1}\int_{0}^{p(x)}yf_{Y|X}(y|x)dyf_{X}(x)dx.
\end{align*}
The function $yf_{Y|X}(y|x)$ is nonnegative and bounded for $y,x\in[0,1]$.
Then by the integral mean value theorem, we have 
\begin{align*}
 & \mathbb{E}|W(\hat{p}_{D}(\textit{data}),F_{Y,X})-W(p_{D}^{*},F_{Y,X})|\\
= & \mathbb{E}\Big|\int_{0}^{1}\int_{p_{D}^{*}(x)}^{\hat{p}_{D}(x;\textit{data})}yf_{Y|X}(y|x)dyf_{X}(x)dx\Big|\\
\leq & \sup_{y,x}|yf_{Y|X}(y|x)|\mathbb{E}\int_{0}^{1}|\hat{p}_{D}(x;\textit{data})-p_{D}^{*}(x)|f_{X}(x)dx.
\end{align*}
The integral on the last line can be decomposed based on the $K$
markets: 
\begin{align*}
\mathbb{E}\int_{0}^{1}|\hat{p}_{D}(x;\textit{data})-p_{D}^{*}(x)|f_{X}(x)dx & \leq\sum_{k=1}^{K}\int_{I_{k}}\left[\mathbb{E}|\hat{p}_{D}(x;\textit{data})-\tilde{p}_{k}|+|\tilde{p}_{k}-p_{D}^{*}(x)|\right]f_{X}(x)dx\\
 & =\sum_{k=1}^{K}\mathbb{E}|\hat{p}_{k}-\tilde{p}_{k}|/K+\sum_{k=1}^{K}\int_{I_{k}}|\tilde{p}_{k}-p_{D}^{*}(x)|f_{X}(x)dx\\
 & \lesssim(K/n)^{1/3}+1/K+\exp\left(-\frac{nc_{1}^{2}}{8K^{2}}+\log K\right)\asymp n^{-1/4},
\end{align*}
where the last line follows from the proof of Theorem \ref{thm:upper-bound-pd}.

For part (ii), since $p_{U}^{*}$ is a scalar, the welfare can be
simplified to 
\begin{align*}
W(p_{U}^{*},F_{Y})=\int_{0}^{p_{U}^{*}}yf_{Y}(y)dy.
\end{align*}
Then we have 
\begin{align*}
\mathbb{E}|W(\hat{p}_{U}(\textit{data}_{Y}),F_{Y})-W(p_{U}^{*},F_{Y})| & =\mathbb{E}\Big|\int_{p_{U}^{*}}^{\hat{p}_{U}(\textit{data}_{Y})}yf_{Y}(y)dy\Big|\\
 & \leq\sup_{y}|yf_{Y}(y)|\mathbb{E}|\hat{p}_{U}(\textit{data}_{Y})-p_{U}^{*}|\\
 & \lesssim n^{-1/3},
\end{align*}
where we have used Corollary  \ref{thm:upper-bound-unif} along with the
fact that $yf_{Y}(y)$ is nonnegative and bounded for $y\in[0,1]$. 
\end{proof}

\section{Proofs for lower bounds }\label{sec: proofs for lower bounds}


\label{sec:proof-lower-x0} 
\begin{proof}[Proof of Theorem \ref{thm:lower-bound-pd-1}]
For Theorem \ref{thm:lower-bound-pd-1}, we use Lemma \ref{lm:lecam-functionals}
to prove the lower bound. Define 
\begin{align*}
\omega_{D}(\epsilon) & \equiv\sup_{F_{1},F_{2}\in\mathcal{F}}\left\{ |p_{D}^{*}(x_{0};F_{1})-p_{D}^{*}(x_{0};F_{2})|\colon H(F_{1}\|F_{2})\leq\epsilon\right\} .
\end{align*}
By Lemma \ref{lm:lecam-functionals}, we have 
\begin{align*}
\inf_{\check{p}_{D} \in \check{\mathcal{D}}}\sup_{F_{Y,X}\in\mathcal{F}}\mathbb{E}_{F_{Y,X}}|\check{p}_{D}(x_{0};\textit{data})-p_{D}^{*}(x_{0})|\geq\frac{1}{8}\omega_{D}\left(1/(2\sqrt{n})\right).
\end{align*}
Therefore, we only need to find a lower bound for $\omega_{D}$. Based
on the explanation in Section \ref{sec:proof-sketch1}, we want to construct two distributions
that are hard to distinguish but their optimal prices are well-separated.
We start by defining two perturbation functions. Let $\phi_{Y}$ be
defined as 
\begin{align}
\begin{split}\phi_{Y}(t)\equiv\begin{cases}
t+1, & t\in[-1,0],\\
-t+1, & t\in[0,2],\\
t-3, & t\in[2,3],\\
0, & \text{ otherwise.}
\end{cases}\end{split}
\label{eq:def-phiY}
\end{align}
Notice that $\phi_{Y}$ is Lipschitz continuous on $\mathbb{R}$.
Let $\phi_{X}$ be defined as 
\begin{align*}
\begin{split}\phi_{X}(t)\equiv\begin{cases}
e^{-(4t-1)^{2}/(1-(4t-1)^{2})}, & t\in(0,1/2),\\
-e^{-(4t-3)^{2}/(1-(4t-3)^{2})}, & t\in(1/2,1),\\
0, & \text{ otherwise.}
\end{cases}\end{split}
\end{align*}
Notice that $\phi_{X}$ is infinitely differentiable on $\mathbb{R}$.
We plot the two perturbation functions in Figure \ref{fig:perturbs}.
\begin{figure}[!htbp]
\caption{Perturbation functions $\phi_{Y}$ and $\phi_{X}$.}
\centering \begin{tikzpicture}
				
				\begin{axis}[
					legend style={nodes={scale=0.8, transform shape}},
					axis y line=center, 
					axis x line=bottom,
					ytick={-1,0,1},
					xtick={-1,0,1,2,3},
					xmin=-2,     xmax=4,
					ymin=-1.2,     ymax=1.2,
					]
					\addplot [
					thick,
					domain=-2:-1, 
					samples=100, 
					color=black,
					]
					{0};
					\addplot [
					thick,
					domain=-1:0, 
					samples=100, 
					color=black,
					]
					{x+1};
					\addplot [
					thick,
					domain=0:2, 
					samples=100, 
					color=black,
					]
					{-x+1};
					\addplot [
					thick,
					domain=2:3, 
					samples=100, 
					color=black,
					]
					{x-3};
					\addplot [
					thick,
					domain=3:4, 
					samples=100, 
					color=black,
					]
					{0};
					\addlegendentry{$\phi_Y$};
					
				\end{axis}
			\end{tikzpicture} \begin{tikzpicture}
				
				\begin{axis}[
					legend style={nodes={scale=0.8, transform shape}},
					axis y line=center, 
					axis x line=bottom,
					ytick={-1,0,1},
					xtick={0,1/2,1},
					xmin=-0.2,     xmax=1.2,
					ymin=-1.2,     ymax=1.2,
					]
					\addplot [
					thick,
					domain=-0.2:0, 
					samples=100, 
					color=black,
					]
					{0};
					\addplot [
					thick,
					domain=0:1/2, 
					samples=100, 
					color=black,
					]
					{exp(-(4*x-1)^2/(1-(4*x-1)^2))};
					\addplot [
					thick,
					domain=1/2:1, 
					samples=100, 
					color=black,
					]
					{-exp(-(4*x-3)^2/(1-(4*x-3)^2))};
					\addplot [
					thick,
					domain=1:1.2, 
					samples=100, 
					color=black,
					]
					{0};
					\addlegendentry{$\phi_X$};
					
				\end{axis}

			\end{tikzpicture} 
\label{fig:perturbs} 
\end{figure}
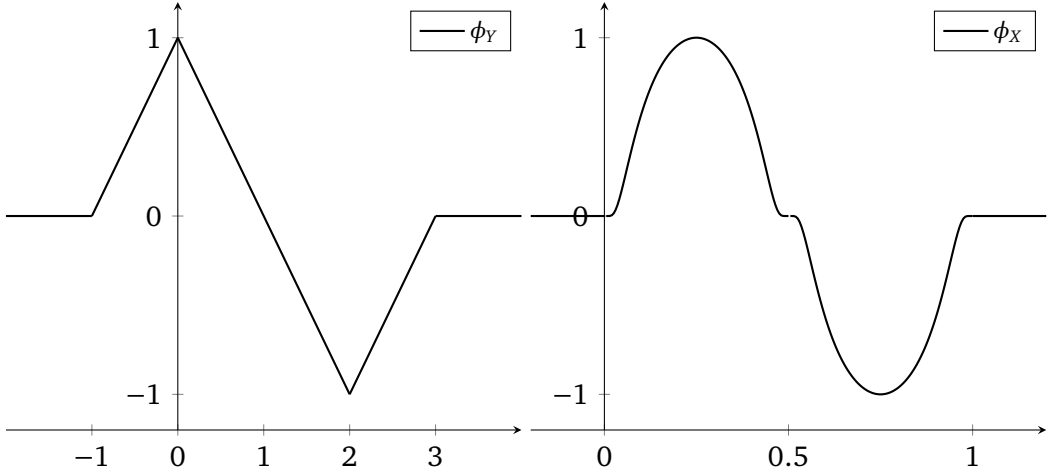

Now we construct the two distributions. Let $\delta\in(0,1/4)$ be
a small number (that depends on $n$) to be specified later. Let $a$
be any number in the interval $(0,4-2C^{*})$. Define the two conditional
density functions of $Y$ given $X$ as 
\begin{align}
f_{1}(y|x) & \equiv1,\nonumber \\
f_{2}(y|x) & \equiv1+a\delta\phi_{Y}\left(\frac{y-1/2}{\delta}\right)\phi_{X}\left(\frac{x-x_{0}}{\delta}+1/4\right).\label{eq:f2(y|x)}
\end{align}
We let the marginal distribution $f_{X}(x)$ of $X$ be the uniform
distribution on $[0,1]$. Note that $f_{1}(y|x)$, $f_{2}(y|x)$,
$f_{1}(y,x)=f_{1}(y|x)f_{X}(x)$, and $f_{2}(y,x)=f_{2}(y|x)f_{X}(x)$
are non-negative everywhere, with integrals over their respective
entire spaces all equaling to $1$.

The first task is to verify that the two distributions are indeed
in the class $\mathcal{F}_{\kappa}$. For $C^{*}\in(0,\,2)$, the
first distribution is in $\mathcal{F}$ by Lemma \ref{lm:uniform-distribution}
and the fact that $Y$ is independent of $X$. Given any $x\in[0,1]$,
we can treat the whole term $a\phi_{X}((x-x_{0})/\delta+1/4)$ as
the coefficient $b$ in Lemma \ref{lm:perturbed-distribution}. Then
the results of Lemma \ref{lm:perturbed-distribution} applies since
$|\phi_{X}|\leq1$. In particular, the revenue function at $x$ is
twice-differentiable a.e., the absolute value of the second-order
partial derivative with respect to $y$ is bounded, and is also bounded
from below by $C^{*}$. The optimal price is an interior solution
and is in the interior of a region on which the revenue function is
twice-differentiable. Lastly, the absolute value of the partial
derivative of $f_{2}(y|x)$ with respect to $x$ is bounded. This
ensures that the quantity $|\frac{\partial}{\partial x}F_{Y|X}(y|x)+y\frac{\partial}{\partial x}f_{Y|X}(y|x)|$
is bounded.

Next, we want to derive the Hellinger distance between the two joint
densities 
\begin{align*}
f_{1}(y,x) & =1,\\
f_{2}(y,x) & =1+a\delta\phi_{Y}\left(\frac{y-1/2}{\delta}\right)\phi_{X}\left(\frac{x-x_{0}}{\delta}+1/4\right).
\end{align*}
Let $\Psi(t)\equiv\sqrt{1+t}$. Its second-order derivate is bounded
when $|t|<1/2$; that is, 
\begin{align*}
\sup_{|t|<1/2}|\Psi''(t)|<C.
\end{align*}
We use $H$ to denote the Hellinger distance: 
\begin{align*}
H(f_{1}\|f_{2})^{2}\equiv\int_{0}^{1}\left(\sqrt{f_{1}(y)}-\sqrt{f_{2}(y)}\right)^{2}dy.
\end{align*}
The Hellinger distance can be bounded as 
\begin{align*}
H^{2}(f_{1}\|f_{2})/2 & =1-\int_{0}^{1}\int_{0}^{1}\Psi\left(a\delta\phi_{Y}\left(\frac{y-1/2}{\delta}\right)\phi_{X}\left(\frac{x-x_{0}}{\delta}+1/4\right)\right)dydx\\
 & =\int_{0}^{1}\int_{0}^{1}\Psi(0)-\Psi\left(a\delta\phi_{Y}\left(\frac{y-1/2}{\delta}\right)\phi_{X}\left(\frac{x-x_{0}}{\delta}+1/4\right)\right)dydx\\
 & \leq-a\Psi'(0)\delta\int_{0}^{1}\int_{0}^{1}\phi_{Y}\left(\frac{y-1/2}{\delta}\right)\phi_{X}\left(\frac{x-x_{0}}{\delta}+1/4\right)dydx\\
 & \quad+a^{2}C\delta^{2}\int_{0}^{1}\int_{0}^{1}\phi_{Y}^{2}\left(\frac{y-1/2}{\delta}\right)\phi_{X}^{2}\left(\frac{x-x_{0}}{\delta}+1/4\right)dydx,
\end{align*}
where we have applied the second-order Taylor expansion to obtain
the last inequality. By the change of variables $u=(y-1/2)/\delta$
and $v=(x-x_{0})/\delta+1/4$, for sufficiently small $\delta\in(0,\,1/2]$,
\begin{align}
 & \int_{0}^{1}\int_{0}^{1}\phi_{Y}\left(\frac{y-1/2}{\delta}\right)\phi_{X}\left(\frac{x-x_{0}}{\delta}+1/4\right)dydx=\delta^{2}\int_{-1}^{1}\phi_{Y}\left(u\right)du\int_{0}^{1}\phi_{X}\left(v\right)dv=0,\label{eq:10}
\end{align}
and 
\begin{align*}
 & \int_{0}^{1}\int_{0}^{1}\phi_{Y}^{2}\left(\frac{y-1/2}{\delta}\right)\phi_{X}^{2}\left(\frac{x-x_{0}}{\delta}+1/4\right)dydx=\delta^{2}\int_{-1}^{1}\phi_{Y}^{2}\left(u\right)du\int_{0}^{1}\phi_{X}^{2}\left(v\right)dv\leq C\delta^{2}.
\end{align*}
Therefore, the Hellinger distance is bounded as 
\begin{align*}
H^{2}(f_{1}\|f_{2})\lesssim\delta^{4}.
\end{align*}
Now we take $\delta$ such that $\delta^{4}\asymp1/n$. Note that
(\ref{eq:10}) holds when $\delta>0$ is small enough such that $\delta\in(0,\,1/2]$,
$1/4-x_{0}/\delta\leq0$ and $(1-x_{0})/\delta+1/4\geq1$; that is,
when $x_{0}n^{1/4}\geq c'$ and $(1-x_{0})n^{1/4}\geq c''$
for positive universal constants $c'$ and $c''$ (independent
of $n$ and $x_{0}$). This ensures that $H^{2}(f_{1}\|f_{2})\lesssim1/n$.
Then from Lemma \ref{lm:lecam-functionals}, we know that 
\begin{align*}
\inf_{\check{p}_{D} \in \check{\mathcal{D}}}\sup_{F_{Y,X}\in\mathcal{F}}\mathbb{E}_{F_{Y,X}}|\check{p}_{D}(x_{0};\textit{data})-p_{D}^{*}(x_{0})|\gtrsim n^{-1/4},x_{0}\in(0,1).
\end{align*}
For bounding the revenue, recall that the revenue achieved at the
price $p$ and covariate value $x_{0}$ is $r(p,x_{0})=\max_{p}p(1-F_{Y|X}(p|x_{0}))$.
By Lemma \ref{lm:foundamental-thm-calculus-Lipschitz}, we have 
\begin{align*}
r(p_{D}^{*}(x_{0}))-r(\check{p}_{D}(x_{0};\textit{data}))\geq\frac{C^{*}}{2}|p_{D}^{*}(x_{0})-\check{p}_{D}(x_{0};\textit{data})|^{2}.
\end{align*}
As a result, we have 
\begin{align*}
 & \inf_{\check{p}_{D} \in \check{\mathcal{D}}}\sup_{F_{Y,X}\in\mathcal{F}}\mathbb{E}[r(p_{D}^{*},x_{0})-r(\check{p}_{D}(\textit{data}),x_{0})]\\
\geq & \inf_{\check{p}_{D} \in \check{\mathcal{D}}}\sup_{F_{Y,X}\in\mathcal{F}}\mathbb{E}\left[\frac{C^{*}}{2}|p_{D}^{*}(x_{0})-\check{p}_{D}(x_{0};\textit{data})|^{2}\right]\\
\geq & \inf_{\check{p}_{D} \in \check{\mathcal{D}}}\sup_{F_{Y,X}\in\mathcal{F}}\frac{C^{*}}{2}\left\{ \mathbb{E}\left[|p_{D}^{*}(x_{0})-\check{p}_{D}(x_{0};\textit{data})|\right]\right\} ^{2}\gtrsim n^{-1/2}.
\end{align*}
This proves Theorem \ref{thm:lower-bound-pd-1}.
\end{proof}
\label{sec:proof-lower-expected} 
\begin{proof}[Proof of Theorem \ref{thm:lower-bound-pd}]
To prove Theorem \ref{thm:lower-bound-pd}, we follow the explanation
in Section \ref{sec:proof-sketch2} and use the Fano's inequality to bound the probability
of mistakes in the multiple classification problem. Before solving
the revenue problem, we first study the lower bound for the $L_{2}-$distance
of pricing functions. For two pricing functions $p_{1}$ and $p_{2}$,
we define the (unweighted) $L_{2}-$distance as 
\begin{align*}
\lVert p_{1}-p_{2}\rVert_{2}\equiv\left(\int_{0}^{1}|p_{1}(x)-p_{2}(x)|^{2}dx\right)^{1/2}.
\end{align*}
In part (i), we defined the perturbation on the $X$ dimension at
a fixed point $x_{0}$. Now we want to define a large set of perturbed
distributions. Each of these distributions is perturbed in a small
interval on the $X$ dimension. Let $m\geq8$ be a large number (depending
on $n$) that we specify later. Let $\alpha\in\{0,1\}^{m}$ be a vector
of length $m$; that is, 
\begin{align*}
\alpha\equiv(\alpha_{1},\ldots,\alpha_{m}),\text{ where }\alpha_{j}\in\{0,1\},j=1,\ldots,m.
\end{align*}
We construct a set of conditional density functions indexed by $\alpha$:
\begin{align*}
f_{Y|X}^{\alpha}(y|x)\equiv1+\frac{a}{m}\sum_{j=1}^{m}\alpha_{j}\phi_{Y}\left(m(y-1/2)\right)\phi_{X}\left(mx-(j-1)\right).
\end{align*}
The marginal distribution of $X$ is taken to be the uniform distribution
on $[0,1]$, that is, $f_{X}\equiv\mathbf{1}_{[0,1]}$. We denote
the joint distribution by $f_{Y,X}^{\alpha}\equiv f_{Y|X}^{\alpha}f_{X}$.

We briefly describe this construction of the conditional density.
The unit interval $[0,1]$ is divided equally into $m$ subintervals:
\begin{align*}
I_{j}\equiv\left[(j-1)/m,j/m\right],j=1,\ldots,m.
\end{align*}
For $x\in I_{j}$, if $\alpha_{j}=0$, then the conditional density
is $1$. If $\alpha_{j}=1$, then the conditional density 
\begin{align*}
f_{Y|X}^{\alpha}(y|x)\equiv1+\frac{a}{m}\phi_{Y}\left(m(y-1/2)\right)\phi_{X}\left(mx-(j-1)\right),x\in I_{j}.
\end{align*}
By treating $1/m$ as the scalar $\delta$ in part (i), we can see
that, for $m$ large enough, each $f_{Y,X}^{\alpha}$ belongs to the
set $\mathcal{F}_{\kappa}$.

From the set $\{f_{Y,X}^{\alpha}:\alpha\in\{0,1\}^{m}\}$, we want
to pick out a large enough subset of distributions whose optimal price
functions are well-separated. For this purpose, we use the Gilbert-Varshamov
bound \citep[Lemma 2.9, Chapter 2][]{tsybakov2009introduction}. The
Gilbert-Varshamov bound states that for $m\geq8$, there exists
a subset $\mathcal{A}\subset\{0,1\}^{m}$ with cardinality $M\equiv|\mathcal{A}|\geq2^{m/8}$,
and the pairwise rescaled Hamming distance between elements in this
set is greater than $1/8$. That is, 
\begin{align*}
\frac{1}{m}\sum_{j=1}^{m}\mathbf{1}\{\alpha_{j}\ne\alpha_{j}'\}\geq\frac{1}{8},\text{ for any }\alpha,\alpha'\in\mathcal{A}.
\end{align*}
Applying the Gilbert-Varshamov bound, we can show that for $\alpha,\alpha'\in\mathcal{A}$,
the optimal pricing functions of $f_{Y,X}^{\alpha}$ and $f_{Y,X}^{\alpha'}$
are well-separated. Let $p_{\alpha}$ be the pricing function associated
with $f_{Y,X}^{\alpha}$; that is, 
\begin{align*}
p_{\alpha}(x)\equiv\argmax_{p\in[0,1]}p(1-F_{Y|X}^{\alpha}(p|x)),
\end{align*}
where $F_{Y|X}^{\alpha}(y|x)$ is the corresponding conditional cumulative
distribution function. Note that $\alpha,\alpha'\in\mathcal{A}$ differ
in at least $m/8$ positions. This means that $f_{Y|X}^{\alpha}$
and $f_{Y|X}^{\alpha'}$ differ in $m/8$ intervals. Suppose that
$I_{j}$ is such an interval, where $\alpha_{j}=0$ and $\alpha{}_{j}'=1$.
We restrict our attention to a subset of this interval: 
\begin{align*}
\tilde{I}_{j}\equiv\left[\frac{1}{6m}+\frac{j-1}{m},\frac{1}{3m}+\frac{j-1}{m}\right]\subset I_{j}.
\end{align*}
When $x\in\tilde{I}_{j}$, we have 
\begin{align}
mx-(j-1)\in\left[1/6,1/3\right]\implies\phi_{X}\left(mx-(j-1)\right)\in\left[\phi_{X}(0),\,\phi_{X}(1/2)\right].\label{eq:11}
\end{align}
By Lemma \ref{lm:perturbed-distribution} (where $b=a\phi_{X}\left(mx-(j-1)\right)>0$,
$\delta=1/m$), the choice $a\in(0,4-2\kappa)$, and the fact (\ref{eq:11}),
if we fix $x\in\tilde{I}_{j}$, then $p_{\alpha}(x)=1/2$ while 
\begin{align*}
p_{\alpha'}(x)\leq1/2-\frac{c}{m}\phi_{X}\left(mx-(j-1)\right)\leq1/2-\frac{c\phi_{X}(1/6)}{m},x\in\tilde{I}_{j},
\end{align*}
where $c>0$ is a universal constant that does not depend on $n$.\footnote{For example, $c$ can be equal to $a/8$ according to Lemma \ref{lm:perturbed-distribution}.}
This implies that 
\begin{align*}
|p_{\alpha}(x)-p_{\alpha'}(x)|\gtrsim\frac{1}{m},x\in\tilde{I}_{j}.
\end{align*}
Therefore, on the interval $I_{j}$, the separation between $p_{\alpha}$
and $p_{\alpha'}$ is lower bounded as 
\begin{align*}
\int_{I_{j}}|p_{\alpha}(x)-p_{\alpha'}(x)|^{2}dx\gtrsim\int_{\tilde{I}_{j}}1/m^{2}dx=\frac{1}{6m}\times\frac{1}{m^{2}}\gtrsim1/m^{3}.
\end{align*}
By the Gilbert-Varshamov bound, there are at least $m/8$ such intervals.
Therefore, we can lower bound the total separation by 
\begin{align*}
\lVert p_{1}-p_{2}\rVert_{2}\gtrsim\left(m/8\times1/m^{3}\right)^{1/2}\gtrsim1/m.
\end{align*}

Next, we want to compute the KL divergence between $f_{Y,X}^{\alpha}$
and $f_{Y,X}^{\alpha'}$. Note that the term $\phi_{X}\left(mx-(j-1)\right)$
is non-zero only when $x\in I_{j}$. The KL divergence can therefore
be treated as a sum of $m$ integrals: 
\begin{align*}
\text{KL}(f_{Y,X}^{\alpha}\|f_{Y,X}^{\alpha'}) & =\int_{0}^{1}\int_{0}^{1}f_{Y,X}^{\alpha}(y,x)\log\frac{f_{Y,X}^{\alpha}}{f_{Y,X}^{\alpha'}}dydx=\sum_{j=1}^{m}E_{j},
\end{align*}
where 
\begin{align*}
E_{j} & \equiv\int_{I_{j}}\int_{0}^{1}\left(1+\frac{a}{m}\alpha_{j}\phi_{Y}\left(m(y-1/2)\right)\phi_{X}\left(mx-(j-1)\right)\right)\\
 & \quad\times\log\frac{1+\frac{a}{m}\alpha_{j}\phi_{Y}\left(m(y-1/2)\right)\phi_{X}\left(mx-(j-1)\right)}{1+\frac{a}{m}\alpha{}_{j}'\phi_{Y}\left(m(y-1/2)\right)\phi_{X}\left(mx-(j-1)\right)}dydx.
\end{align*}
Notice that when $\alpha_{j}=\alpha_{j'}$, $E_{j}=0$. Therefore,
we only need to consider the $j$'s where $\alpha_{j}\ne\alpha{}_{j}'$.
Denote $\Psi_{1}(t)=-\log(1+t)$ and $\Psi_{2}(t)=(1+t)\log(1+t)$.
Then we can write $E_{j}$ as 
\begin{align*}
E_{j}=\begin{cases}
\int_{I_{j}}\int_{0}^{1}\Psi_{1}\left(\frac{a}{m}\phi_{Y}\left(m(y-1/2)\right)\phi_{X}\left(mx-(j-1)\right)\right)dydx, & \text{ if }\alpha_{j}=0,\alpha{}_{j}'=1,\\
\int_{I_{j}}\int_{0}^{1}\Psi_{2}\left(\frac{a}{m}\phi_{Y}\left(m(y-1/2)\right)\phi_{X}\left(mx-(j-1)\right)\right)dydx, & \text{ if }\alpha_{j}=1,\alpha{}_{j}'=0.
\end{cases}
\end{align*}
By the second-order Taylor expansion at zero, we have 
\begin{align*}
\Psi_{1}(t)=-t+\frac{1}{2(1+t')^{2}}t^{2},
\end{align*}
for some $t'$ between $0$ and $t$. When $|t|\leq1/4$,\footnote{Later we show that $m$ is chosen to be $c_{0}n^{1/4}$ where $c_{0}>0$
is a universal constant. As a result, $|t|\leq1/4$ is guaranteed
as long as $c_{0}$ is sufficiently large.} we have 
\begin{align*}
\Psi_{1}(t)\leq-t+Ct^{2},
\end{align*}
for some universal constant $C>0$. Similarly, we can show that 
\begin{align*}
\Psi_{2}(t)\leq t+Ct^{2}.
\end{align*}
Applying these inequalities to $E_{j}$, we have 
\begin{align*}
E_{j}\leq & \pm\int_{I_{j}}\int_{0}^{1}\frac{a}{m}\phi_{Y}\left(m(y-1/2)\right)\phi_{X}\left(mx-(j-1)\right)dydx\\
 & +C\int_{I_{j}}\int_{0}^{1}\frac{a^{2}}{m^{2}}\phi_{Y}^{2}\left(m(y-1/2)\right)\phi_{X}^{2}\left(mx-(j-1)\right)dydx.
\end{align*}
Similar to the derivation in Part (i), we know that the first term
on the RHS is zero. For the second term, we can apply change of variables
$u=m(y-1/2)$ and $v=mx-(j-1)$ and obtain that 
\begin{align*}
 & \int_{I_{j}}\int_{0}^{1}\phi_{Y}^{2}\left(m(y-1/2)\right)\phi_{X}^{2}\left(mx-(j-1)\right)dydx\\
= & \frac{1}{m^{2}}\int_{0}^{1}\phi_{X}^{2}\left(v\right)dv\int_{-1}^{3}\phi_{Y}^{2}\left(u\right)du\leq\frac{C'}{m^{2}}
\end{align*}
for some universal constant $C'>0$. Putting the results results
together, we know that $E_{j}\leq\frac{C}{m^{4}}$ for all $j$. Since
there are $m$ intervals, we can bound the KL divergence by 
\begin{align*}
\text{KL}(f_{Y,X}^{\alpha}\|f_{Y,X}^{\alpha'})=\sum_{j=1}^{m}E_{j}\lesssim\frac{1}{m^{3}}.
\end{align*}
This is the KL distance for a single observation. For the entire data
set with $n$ i.i.d. observations, the KL divergence is upper bounded
by $Cn/m^{3}$.

Lastly, we can summarize our results into the Fano inequality presented
in Lemma \ref{lm:fano}. We have 
\begin{align*}
\inf_{\check{p}_{D} \in \check{\mathcal{D}}}\sup_{F_{Y,X}\in\mathcal{F}}\mathbb{E}\lVert\check{p}_{D}(\textit{data})-p_{D}^{*}\rVert_{2}^{2} & \geq\frac{C_{1}}{m^{2}}\left(1-\frac{C_{2}n/m^{3}+\log2}{\log2^{m/8}}\right)\\
 & \geq\frac{C_{1}}{m^{2}}\left(1-\frac{C_{2}n/m^{3}+\log2}{C_{3}m}\right).
\end{align*}
By choosing $m=c_{0}n^{1/4}$ for a sufficiently large universal constant
$c_{0}>0$, we can make the factor $\left(1-\frac{C_{2}n/m^{3}+\log2}{C_{3}m}\right)$
stay above, say, $1/2$. Then we have 
\begin{align*}
\inf_{\check{p}_{D} \in \check{\mathcal{D}}}\sup_{F_{Y,X}\in\mathcal{F}}\mathbb{E}\lVert\check{p}_{D}(\textit{data})-p_{D}^{*}\rVert_{2}^{2}\gtrsim\frac{1}{m^{2}}\asymp n^{-1/2}.
\end{align*}

So far we have derived the lower bound for the $L_{2}-$distance of
pricing. Moving onto the revenue problem, recall that the revenue
achieved at the price $p$ and covariate value $x$ is $r(p,x)=\max_{p}p(1-F_{Y|X}(p|x))$.
By Lemma \ref{lm:foundamental-thm-calculus-Lipschitz}, we have 
\begin{align*}
r(p_{D}^{*},x)-r(\check{p}_{D}(\textit{data}),x)\geq\frac{C^{*}}{2}|p_{D}^{*}(x)-\check{p}_{D}(x;\textit{data})|^{2}.
\end{align*}
Since $f_{X}$ is bounded away from zero, we have 
\begin{align*}
 & \inf_{\check{p}_{D} \in \check{\mathcal{D}}}\sup_{F_{Y,X}\in\mathcal{F}}\mathbb{E}[R(p^{*}_D)-R(\check{p}_D)]\\
= & \inf_{\check{p}_{D} \in \check{\mathcal{D}}}\sup_{F_{Y,X}\in\mathcal{F}}\mathbb{E}\left[\int_{0}^{1}(r(p^{*}_D,x)-r(\check{p}_D,x))f_{X}(x)dx\right]\\
\geq & \inf_{\check{p}_{D} \in \check{\mathcal{D}}}\sup_{F_{Y,X}\in\mathcal{F}}\mathbb{E}\left[\frac{C^{*}}{2} \left(\inf_{x\in[0,1]}f_{X}(x)\right)\int_{0}^{1}|p_{D}^{*}(x)-\check{p}_{D}(x;\textit{data})|^{2}dx\right]\gtrsim n^{-1/2}.
\end{align*}
\end{proof}
\label{sec:proof-lower-uniform} 
\begin{proof}[Proof of Theorem \ref{thm:lower-bound-uniform}]
We use Lemma \ref{lm:lecam-functionals} to prove the lower bound
for Theorem \ref{thm:lower-bound-uniform}. Define 
\begin{align*}
\omega_{U}(\epsilon) & \equiv\sup_{F_{1},F_{2}\in\mathcal{F}^{U}}\left\{ |p_{U}^{*}(F_{1})-p_{U}^{*}(F_{2})|:H(F_{1}\|F_{2})\leq\epsilon\right\} .
\end{align*}
Then by Lemma \ref{lm:lecam-functionals}, we have 
\begin{align*}
\inf_{\check{p}_{U}\in\check{\mathcal{U}}}\sup_{F_{Y}\in\mathcal{F}^{U}}\mathbb{E}_{F_{Y}}|\check{p}_{U}(\textit{data}_{Y})-p_{U}^{*}|\geq\frac{1}{8}\omega_{U}\left(1/(2\sqrt{n})\right).
\end{align*}
Therefore, we only need to find a lower bound for $\omega_{U}$. The
proof proceeds in three steps. In the first step, we construct two
distributions and compute the separation between their optimal prices.
The second step bounds the Hellinger distance between these two distributions.
The third step summarizes.

\bigskip{}
\textbf{Step 1.} We construct two distribution functions. The first
distribution is the uniform distribution on the unit interval $[0,1]$.
We denote this density function as 
\begin{align*}
f_{1}(y)=1_{[0,1]}(y).
\end{align*}
The distribution function is $F_{1}(y)=y$ on the support $[0,1]$.
The revenue function under this distribution is $R_{1}(p)=p(1-p)$.
The optimal price is 
\begin{align*}
p_{1}=\argmax_{p\in[0,1]}R_{1}(p)=\argmax_{p\in[0,1]}p-p^{2}=1/2.
\end{align*}

\noindent The second distribution function is a small twist of the
uniform distribution. We use the same perturbation function $\phi_{Y}$
defined in (\ref{eq:def-phiY}).

We apply a small perturbation to the uniform density. Let $\delta>0$
be a small number (that depends on $n$) specified later. Let $a\in(0,4-2C^{*})$.
The formula of the density $f_{2}$ is given by 
\begin{align*}
f_{2}(y)\equiv1+a\delta\phi_{Y}\left(\frac{y-1/2}{\delta}\right)=\begin{cases}
1, & \text{ if }y\in[0,1/2-\delta),\\
ay+1-\frac{a}{2}+a\delta, & \text{ if }y\in[1/2-\delta,1/2),\\
-ay+1+\frac{a}{2}+a\delta, & \text{ if }y\in[1/2,1/2+2\delta),\\
ay+1-\frac{a}{2}-3a\delta, & \text{ if }y\in[1/2+2\delta,1/2+3\delta),\\
1, & \text{ if }y\in[1/2+3\delta,1].
\end{cases}
\end{align*}
We compare the two densities $f_{1}$ and $f_{2}$ in the following
graph. 
 
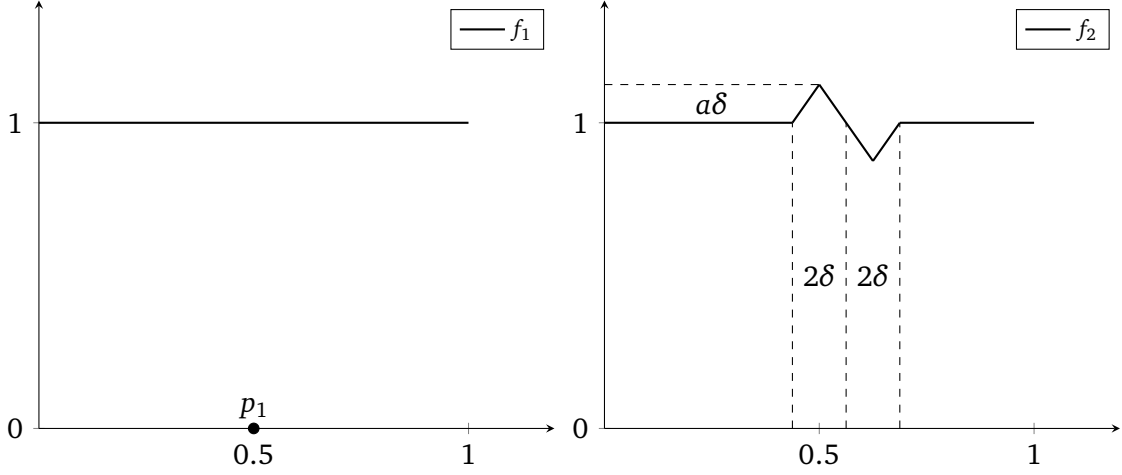
\begin{figure}[!htbp]
\caption{Density functions $f_1$ and $f_2$.}
\centering \begin{tikzpicture}
				
				\begin{axis}[
					legend style={nodes={scale=0.8, transform shape}},
					axis y line=left, 
					axis x line=bottom,
					ytick={0,1},
					xtick={0.5,1},
					xmin=0,     xmax=1.2,
					ymin=0,     ymax=1.4,
					]
					\addplot [
					thick,
					domain=0:1, 
					samples=100, 
					color=black,
					]
					{1};
					
					\addlegendentry{$f_1$};
					\addplot[mark=*] coordinates{(0.5,0)} node[above] {$p_1$};

					
				\end{axis}

			\end{tikzpicture} \begin{tikzpicture}
				
				\begin{axis}[
					legend style={nodes={scale=0.8, transform shape}},
					axis y line=left, 
					axis x line=bottom,
					ytick={0,1},
					xtick={0.5,1},
					xmin=0,     xmax=1.2,
					ymin=0,     ymax=1.4,
					]
					\addplot [
					thick,
					domain=0:7/16, 
					samples=100, 
					color=black,
					]
					{1};
					\addplot [
					thick,
					domain=7/16:1/2, 
					samples=100, 
					color=black,
					]
					{2*x + 2/16};
					\addplot [
					thick,
					domain=1/2:5/8, 
					samples=100, 
					color=black,
					]
					{-2*x + 2 + 2/16};
					\addplot [
					thick,
					domain=5/8:11/16, 
					samples=100, 
					color=black,
					]
					{2*x -6/16};
					\addplot [
					thick,
					domain=11/16:1, 
					samples=100, 
					color=black,
					]
					{1};
					\addlegendentry{$f_2$};
					
					\addplot[dashed] coordinates{(0.5-1/16,0) (0.5-1/16,1)};
					\addplot[dashed] coordinates{(0.5+1/16,0) (0.5+1/16,1)};
					\addplot[] coordinates{(0.5,0.5)} node[] {$2\delta$};
					\addplot[dashed] coordinates{(0.5+3/16,0) (0.5+3/16,1)};
					\addplot[] coordinates{(0.5+2/16,0.5)} node[] {$2\delta$};
					
					\addplot[dashed] coordinates{(0,1+2/16) (1/2 ,1+2/16)};
					\addplot[] coordinates{(0.25,1+1/16)} node[] {$a\delta$};

					
				\end{axis}

			\end{tikzpicture} 
\end{figure}

Denote the optimal price under $f_{2}$ by $p_{2}$. By Lemma \ref{lm:perturbed-distribution}(ii),
we have 
\begin{align*}
|p_{2}-p_{1}|\geq a\delta/8
\end{align*}
when $\delta$ is sufficiently small.

\bigskip{}
\textbf{Step 2.} We want to bound the Hellinger distance $H(F_{1}\|F_{2})$.
Define the function $\Psi(t)=\sqrt{1+t}$. Its second-order derivative
is bounded when $|t|<1/2$; that is, 
\begin{align*}
\sup_{|t|<1/2}|\Psi''(t)|\leq\frac{\sqrt{2}}{2}.
\end{align*}
Since $f_{1}(y)=1$, we have 
\begin{align*}
H(F_{1}\|F_{2})^{2}/2 & =1-\int_{0}^{1}\Psi\left(a\delta\phi_{Y}\left(\frac{y-1/2}{\delta}\right)\right)dy\\
 & =\int_{0}^{1}\Psi(0)-\Psi\left(a\delta\phi_{Y}\left(\frac{y-1/2}{\delta}\right)\right)dy.
\end{align*}
By the second-order Taylor expansion, we have 
\begin{align*}
 & \Psi(0)-\Psi\left(a\delta\phi_{Y}\left(\frac{y-1/2}{\delta}\right)\right)\\
\leq & -\Psi'(0)a\delta\phi_{Y}\left(\frac{y-1/2}{\delta}\right)+\frac{\sqrt{2}}{4}a^{2}\delta^{2}\phi_{Y}^{2}\left(\frac{y-1/2}{\delta}\right).
\end{align*}
By the construction of $\phi_{Y}$, we have 
\begin{align*}
\int_{0}^{1}\phi_{Y}\left(\frac{y-1/2}{\delta}\right)dy=0.
\end{align*}
By the change of variables $u=(y-1/2)/\delta$, we have 
\begin{align*}
\int_{0}^{1}\phi_{Y}^{2}\left(\frac{y-1/2}{\delta}\right)dy=\delta\int_{\mathbb{R}}\phi_{Y}^{2}\left(u\right)du\leq4\delta\int_{-1}^{0}(x+1)^{2}dx=\frac{4}{3}\delta.
\end{align*}
Combining these results together, we obtain a bound on the Hellinger
distance 
\begin{align*}
H(F_{1}\|F_{2})^{2}\leq\frac{2\sqrt{2}}{3}a^{2}\delta^{3}.
\end{align*}

\noindent \bigskip{}
\textbf{Step 3.} By setting $\delta=c_{0}'(3/8\sqrt{2})^{1/3}a^{-2/3}n^{-1/3}$
for $c_{0}'\in(0,1)$, we can ensure that $H(F_{1}\|F_{2})\leq1/(2\sqrt{n})$.
Previously, we assumed that $a\delta\leq1/2$ for the second-order
Taylor expansion. This is true if $c_{0}'$ is chosen to be sufficiently
small. In this case, the separation between $p_{1}$ and $p_{2}$
is lower bounded as below: 
\begin{align*}
|p_{1}-p_{2}|\geq a\delta/8=\frac{c_{0}'}{16}\left(\frac{3}{\sqrt{2}}\right)^{1/3}\left(\frac{a}{n}\right)^{1/3}.
\end{align*}
By Lemma \ref{lm:lecam-functionals}, we have 
\begin{align*}
\inf_{\check{p}_{U}\in\check{\mathcal{U}}}\sup_{F_{Y}\in\mathcal{F}^{U}}\mathbb{E}|\check{p}_{U}(\textit{data}_{Y})-p_{U}^{*}|\geq\frac{c_{0}'}{16}\left(\frac{3}{\sqrt{2}}\right)^{1/3}\left(\frac{a}{n}\right)^{1/3}.
\end{align*}

\noindent Lastly, we want to lower bound the revenue. By Lemma \ref{lm:foundamental-thm-calculus-Lipschitz},
we have
\begin{align*}
\mathscr{R}_{n}^{U}(\mathcal{F}^{U}) & =\inf_{\check{p}_{U}\in\check{\mathcal{U}}}\sup_{F_{Y}\in\mathcal{F}^{U}}\mathbb{E}|R(\check{p}_{U}(\textit{data}_{Y}),F_{Y})-R(p_{U}^{*},F_{Y})|\\
 & \geq\inf_{\check{p}_{U}\in\check{\mathcal{U}}}\sup_{F_{Y}\in\mathcal{F}^{U}}\mathbb{E}\left[\frac{C^{*}}{2}|\check{p}_{U}(\textit{data}_{Y})-p_{U}^{*}|^{2}\right]\\
 & \geq\inf_{\check{p}_{U}\in\check{\mathcal{U}}}\sup_{F_{Y}\in\mathcal{F}^{U}}\frac{C^{*}}{2}\left\{ \mathbb{E}\left[|\check{p}_{U}(\textit{data}_{Y})-p_{U}^{*}|\right]\right\} ^{2}\\
 & \gtrsim\left(\frac{1}{n}\right)^{2/3}.
\end{align*}
\end{proof}

\section{Auxiliary Lemmas}

\begin{lemma} \label{lm:foundamental-thm-calculus-Lipschitz} Let
$f$ be a function on $[0,1]$. Assume that $f$ is differentiable
and its derivative $f'$ is Lipschitz continuous. Let $z^{*}$ be
a point in $[0,1]$ such that $f'(z^{*})=0$.

\begin{enumerate}[label = (\roman*)]
\item The derivative $f'$ is a.e. differentiable on $[0,1]$. 
\item Assume that there exists $\kappa_{1}>0$ such that $f''(z)\leq-\kappa_{1}$
for almost all $z\in[0,1]$. Then, for any $z\in[0,1]$, we have 
\begin{align*}
|f(z)-f(z^{*})|\geq\frac{\kappa_{1}}{2}(z-z^{*})^{2}.
\end{align*}
\item Assume that there exists $\kappa_{2}>0$ such that $|f''(z)|\leq\kappa_{2}$
for almost all $z\in[0,1]$. Then, for any $z\in[0,1]$, we have 
\begin{align*}
|f(z)-f(z^{*})|\leq\frac{\kappa_{2}}{2}(z-z^{*})^{2}.
\end{align*}
\end{enumerate}
\end{lemma} 
\begin{proof}[Proof of Lemma \ref{lm:foundamental-thm-calculus-Lipschitz}]
For part (i), notice that a Lipschitz continuous function is absolutely
continuous. By Theorem 3.35 in Chapter 3 of \citet{folland1999real},
we know that $f'$ is differentiable a.e. with 
\begin{align*}
f'(z_{1})-f'(z_{2})=\int_{z_{2}}^{z_{1}}f''(z)dz.
\end{align*}
For part (ii), we can apply the fundamental theorem of calculus twice
and obtain that 
\begin{align*}
f(z)-f(z^{*}) & =\int_{z^{*}}^{z}f'(\tilde{z})d\tilde{z}\\
 & =\int_{z^{*}}^{z}(f'(z_{1})-f'(z^{*}))dz_{1}\\
 & =\int_{z^{*}}^{z}\int_{z^{*}}^{z_{1}}f''(z_{2})dz_{2}dz_{1}\\
 & \leq-\kappa_{1}\int_{z^{*}}^{z}\int_{z^{*}}^{z_{1}}dz_{2}dz_{1},
\end{align*}
where in the second line we have used the assumption that $f'(z^{*})=0$,
and in the last line we have used the assumption that $f''(z)\leq-\kappa_{1}$
for almost all $z\in[0,1]$. The double integral in the last line
is equal to 
\begin{align*}
\int_{z^{*}}^{z}\int_{z^{*}}^{z_{1}}dz_{2}dz_{1}=\int_{z^{*}}^{z}(z_{1}-z^{*})dz_{1}=\frac{(z-z^{*})^{2}}{2}.
\end{align*}
Therefore, we have 
\begin{align*}
|f(z)-f(z^{*})|\geq\frac{\kappa_{1}}{2}(z-z^{*})^{2}.
\end{align*}
Part (iii) can be proved analogously. 
\end{proof}
\begin{lemma} \label{lm:uniform-distribution} For the uniform distribution
on $[0,1]$, the revenue function $R(y)=y(1-y)$. The revenue function
is twice-differentiable with second-order derivative $R''(y)=-2,y\in[0,1]$.
The optimal price is $1/2$. \end{lemma} 
\begin{proof}[Proof of Lemma \ref{lm:uniform-distribution}]
The proof is straightforward. 
\end{proof}
\begin{lemma} \label{lm:perturbed-distribution} Recall the perturbation
function $\phi_{Y}$ defined in (\ref{eq:def-phiY}). Consider the
following density function 
\begin{align*}
f(y)\equiv1+b\delta\phi_{Y}\left(\frac{y-1/2}{\delta}\right)=\begin{cases}
1, & \text{ if }y\in[0,1/2-\delta),\\
by+1-\frac{b}{2}+b\delta, & \text{ if }y\in[1/2-\delta,1/2),\\
-by+1+\frac{b}{2}+b\delta, & \text{ if }y\in[1/2,1/2+2\delta),\\
by+1-\frac{b}{2}-3b\delta, & \text{ if }y\in[1/2+2\delta,1/2+3\delta),\\
1, & \text{ if }y\in[1/2+3\delta,1],\\
0, & \text{ otherwise.}
\end{cases}
\end{align*}
Denote $F$ as the corresponding cumulative distribution function,
$R(y)\equiv y(1-F(y))$ the revenue function, and $p^{*}\equiv\argmax_{y\in[0,1]}R(y)$
the optimal price. If $C^{*}\in(0,2)$, $|b|<4-2C^{*}$, and $\delta>0$
is sufficiently small, then the following statements hold.

\begin{enumerate}[label = (\roman*)]
\item The density $f$ is Lipschitz continuous. 
\item The revenue function is twice-differentiable a.e. The second-order
derivative is bounded a.e. and satisfies that 
\begin{align*}
-2f(y)-yf'(y)\geq-C^{*}\text{ for almost all }y.
\end{align*}
\item For $b>0$, the optimal price $p^{*}\in(1/2-\delta,1/2-b\delta/8)$.
For $b<0$, the optimal price $p^{*}\in(1/2-b\delta/8,1/2+2\delta)$.
For $b=0$, the optimal price $p^{*}=1/2$. In particular, $p^{*}$
is always an interior solution, and $f$ is always differentiable
in a neighborhood of $p^{*}$. 
\end{enumerate}
\end{lemma} 
\begin{proof}[Proof of Lemma \ref{lm:perturbed-distribution}]
\begin{figure}[!htbp]
\caption{Perturbation function and perturbed density.}
\centering \begin{tikzpicture}
				
				\begin{axis}[
					legend style={nodes={scale=0.8, transform shape}},
					axis y line=center, 
					axis x line=bottom,
					ytick={-1,0,1},
					xtick={-1,0,1,2,3},
					xmin=-2,     xmax=4,
					ymin=-1.2,     ymax=1.2,
					]
					\addplot [
					thick,
					domain=-2:-1, 
					samples=100, 
					color=black,
					]
					{0};
					\addplot [
					thick,
					domain=-1:0, 
					samples=100, 
					color=black,
					]
					{x+1};
					\addplot [
					thick,
					domain=0:2, 
					samples=100, 
					color=black,
					]
					{-x+1};
					\addplot [
					thick,
					domain=2:3, 
					samples=100, 
					color=black,
					]
					{x-3};
					\addplot [
					thick,
					domain=3:4, 
					samples=100, 
					color=black,
					]
					{0};
					\addlegendentry{$\phi_Y$};
					
				\end{axis}

			\end{tikzpicture} \begin{tikzpicture}
				
				\begin{axis}[
					legend style={nodes={scale=0.8, transform shape}},
					axis y line=left, 
					axis x line=bottom,
					ytick={0,1},
					xtick={0.5,1},
					xmin=0,     xmax=1.2,
					ymin=0,     ymax=1.4,
					]
					\addplot [
					thick,
					domain=0:7/16, 
					samples=100, 
					color=black,
					]
					{1};
					\addplot [
					thick,
					domain=7/16:1/2, 
					samples=100, 
					color=black,
					]
					{2*x + 2/16};
					\addplot [
					thick,
					domain=1/2:5/8, 
					samples=100, 
					color=black,
					]
					{-2*x + 2 + 2/16};
					\addplot [
					thick,
					domain=5/8:11/16, 
					samples=100, 
					color=black,
					]
					{2*x -6/16};
					\addplot [
					thick,
					domain=11/16:1, 
					samples=100, 
					color=black,
					]
					{1};
					\addlegendentry{$f$};
					
					\addplot[dashed] coordinates{(0.5-1/16,0) (0.5-1/16,1)};
					\addplot[dashed] coordinates{(0.5+1/16,0) (0.5+1/16,1)};
					\addplot[] coordinates{(0.5,0.5)} node[] {$2\delta$};
					\addplot[dashed] coordinates{(0.5+3/16,0) (0.5+3/16,1)};
					\addplot[] coordinates{(0.5+2/16,0.5)} node[] {$2\delta$};
					
					\addplot[dashed] coordinates{(0,1+2/16) (1/2 ,1+2/16)};
					\addplot[] coordinates{(0.25,1+1/16)} node[] {$a\delta$};

					
				\end{axis}

			\end{tikzpicture} 
\end{figure}
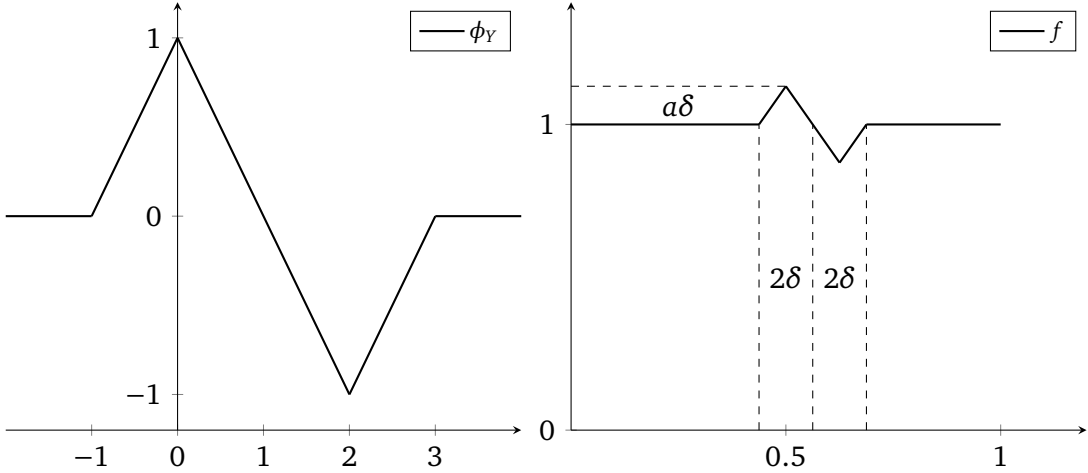

For reference, we plot here the perturbation function $\phi_{Y}$
and the perturbed density $f$. Part (i) is straightforward. The density
$f$ is piecewise linear and hence Lipschitz continuous with Lipschitz
constant $b$. To verify the strong concavity in part (ii), note that
the corresponding revenue function $R$ is continuously differentiable
and twice-differentiable a.e. on the support $[0,1]$. Its second-order
derivative 
\begin{align*}
R''(y)=-2f(y)-yf'(y)=\begin{cases}
-2, & \text{ if }y\in[0,1/2-\delta],\\
-3by-2+b-2b\delta, & \text{ if }y\in[1/2-\delta,1/2],\\
3by-2-b-2b\delta, & \text{ if }y\in[1/2,1/2+2\delta],\\
-3by-2+b+6b\delta, & \text{ if }y\in[1/2+2\delta,1/2+3\delta],\\
-2, & \text{ if }y\in[1/2+3\delta,1].
\end{cases}
\end{align*}
We can see that $R''$ is piecewise linear and hence bounded a.e.
We further show that $R''$ is bounded away from zero by $\kappa$.
On the intervals $[0,1/2-\delta]$ and $[1/2+3\delta,1]$, we have
$R''(y)=-2<-C^{*}$. We check the remaining three intervals one by
one. On the interval $[1/2-\delta,1/2]$, the condition $|b|<4-2C^{*}$
ensures that 
\begin{align*}
b\geq0 & \implies R''(y)\leq R''(1/2-\delta)=-b/2-2+b\delta\leq-C^{*},\\
b<0 & \implies R''(y)\leq R''(1/2)=-b/2-2-2b\delta\leq-C^{*},
\end{align*}
when $\delta$ is sufficiently small. On the interval $[1/2,1/2+2\delta]$,
we have 
\begin{align*}
b\geq0 & \implies R''(y)\leq R''(1/2+2\delta)=b/2-2+4b\delta\leq-C^{*},\\
b<0 & \implies R''(y)\leq R''(1/2)=b/2-2-2b\delta\leq-C^{*},
\end{align*}
when $\delta$ is sufficiently small. On the interval $[1/2+2\delta,1/2+3\delta]$,
we have 
\begin{align*}
b\geq0 & \implies R''(y)\leq R''(1/2+2\delta)=-b/2-2<-C^{*},\\
b<0 & \implies R''(y)\leq R''(1/2+3\delta)=-b/2-2-3b\delta<-C^{*},
\end{align*}
To summarize, we have shown that $R''(y)\leq-C^{*}$ a.e. on $[0,1]$
provided that $\delta>0$ is sufficiently small.

For part (iii), we first consider the case $b>0$. We only need to
consider the interval $[1/2-\delta,1/2]$. The reason will become
clear later. The cumulative distribution function 
\begin{align*}
F(y)=\frac{b}{2}y^{2}+\left(1-\frac{b}{2}+b\delta\right)\delta y+\frac{b}{2}\left(1/2-\delta\right)^{2},y\in\left[1/2-\delta,1/2\right].
\end{align*}
The revenue function 
\begin{align*}
R(y)=-\frac{b}{2}y^{3}-\left(1-\frac{b}{2}+b\delta\right)y^{2}+\left(1-\frac{b}{2}\left(1/2-\delta\right)^{2}\right)y,y\in\left[1/2-\delta,1/2\right].
\end{align*}
The marginal revenue 
\begin{align*}
R'(y)=-\frac{3b}{2}y^{2}-\left(2-b+2b\delta\right)y+1-\frac{b}{2}\left(1/2-\delta\right)^{2},y\in\left[1/2-\delta,1/2\right].
\end{align*}
We evaluate the marginal revenue at two points $1/2-\delta$ and $1/2-\frac{b\delta}{8}$.
When $y=1/2-\delta$, the marginal revenue 
\begin{align*}
R'\left(1/2-\delta\right)=\delta>0.
\end{align*}
When $y=1/2-b\delta/8$, the marginal revenue 
\begin{align*}
R'\left(1/2-\frac{b\delta}{8}\right)\approx\frac{b(b-4)}{16}\delta<0,
\end{align*}
where we have omitted higher order terms involving $\delta^{2}$.
Therefore, $R'\left(1/2-\frac{b\delta}{8}\right)$ is negative for
sufficiently small $\delta$. Since the marginal revenue $R'$ is
strictly decreasing on the entire domain $[0,1]$, we know that the
only zero of $R'$ (which is the optimal price $p^{*}$) is within
the region $(1/2-\delta,1/2-\frac{b\delta}{8})$. Within this region,
the revenue is twice-differentiable everywhere.

Next, we consider the case $b<0$. In this case, we only need to study
the region $[1/2,1/2+2\delta]$. The cumulative distribution function
\begin{align*}
F(y)=-\frac{b}{2}y^{2}+\left(1+\frac{b}{2}+b\delta\right)y+\frac{b}{2}\delta^{2}-\frac{b}{2}\delta-\frac{b}{8},y\in\left[1/2,1/2+2\delta\right].
\end{align*}
The revenue function 
\begin{align*}
R(y)=y(1-F(y))=\frac{b}{2}y^{3}-\left(1+\frac{b}{2}+b\delta\right)y^{2}+\left(1+\frac{b}{8}-\frac{b}{2}\delta^{2}+\frac{b}{2}\delta\right)y,y\in\left[1/2,1/2+2\delta\right].
\end{align*}
The marginal revenue 
\begin{align*}
R'(y)=\frac{3b}{2}y^{2}-(2+b+2b\delta)y+\left(1+\frac{b}{8}-\frac{b}{2}\delta^{2}+\frac{b}{2}\delta\right),y\in\left[1/2,1/2+2\delta\right].
\end{align*}
We evaluate the marginal revenue at two points $1/2+\delta$ and $1/2-\frac{b\delta}{8}$.
When $y=1/2+\delta$, the marginal revenue 
\begin{align*}
R'\left(1/2+\delta\right)\approx-2\delta<0,
\end{align*}
where we have omitted higher order terms involving $\delta^{2}$.
When $y=1/2-b\delta/8$, the marginal revenue 
\begin{align*}
R'\left(1/2-\frac{b\delta}{8}\right)\approx\frac{b(b+4)}{16}\delta>0,
\end{align*}
where we have omitted higher order terms involving $\delta^{2}$.
Since the marginal revenue $R'$ is strictly decreasing on the entire
domain $[0,1]$, we know that the only zero of $R'$ (which is the
optimal price $p^{*}$) is within the region $(1/2-\frac{b\delta}{8},1/2+\delta)$.
Within this region, the revenue is twice-differentiable everywhere.

Lastly, when $b=0$, the density function is constant, and Lemma \ref{lm:uniform-distribution}
shows that the optimal price is $1/2$. Therefore, regardless of the
sign of $b$, the optimal price is always an interior solution, and
is in the interior of a region on which the revenue function is twice-differentiable. 
\end{proof}
\begin{lemma} \label{lm:lecam-functionals} Take $x_{0}\in[0,1]$.
Recall the following definition of $\omega_{D}(\epsilon)$ and $\omega_{U}(\epsilon)$:
\begin{align*}
\omega_{D}(\epsilon) & \equiv\sup_{F_{1},F_{2}\in\mathcal{F}}\left\{ |p_{D}^{*}(x_{0};F_{1})-p_{D}^{*}(x_{0};F_{2})|:H(F_{1}\|F_{2})\leq\epsilon\right\} ,\\
\omega_{U}(\epsilon) & \equiv\sup_{F_{1},F_{2}\in\mathcal{F}^{U}}\left\{ |p_{U}^{*}(F_{1})-p_{U}^{*}(F_{2})|:H(F_{1}\|F_{2})\leq\epsilon\right\} .
\end{align*}
Then 
\begin{align*}
\inf_{\check{p}_{D} \in \check{\mathcal{D}}}\sup_{F_{Y,X}\in\mathcal{F}}\mathbb{E}_{F_{Y,X}}|\check{p}_{D}(x_{0};\textit{data})-p_{D}^{*}(x_{0};F_{Y,X})| & \geq\frac{1}{8}\omega_{D}\left(1/(2\sqrt{n})\right),\\
\inf_{\check{p}_{U}\in\check{\mathcal{U}}}\sup_{F_{Y}\in\mathcal{F}^{U}}\mathbb{E}_{F_{Y}}|\check{p}_{U}(\textit{data}_{Y})-p_{U}^{*}(F_{Y})| & \geq\frac{1}{8}\omega_{U}\left(1/(2\sqrt{n})\right).
\end{align*}
\end{lemma} 
\begin{proof}[Proof of Lemma \ref{lm:lecam-functionals}]
By treating $p_{D}^{*}(x_{0};\cdot)$ and $p_{U}^{*}(\cdot)$ as
functionals, the desired results directly follow from Corollary 15.6
(Le Cam for functionals) in Chapter 15 of \citet{wainwright2019high}. 
\end{proof}
\begin{lemma} \label{lm:fano} Let $\{F_{Y,X}^{j}\colon 1\leq j\leq M\}\subset\mathcal{F}$
be such that 
\begin{align*}
\lVert p_{D}^{*}(F_{Y,X}^{j})-p_{D}^{*}(F_{Y,X}^{j})\rVert_{2}\geq2\delta,j\ne j'.
\end{align*}
Then we have 
\begin{align*}
\inf_{\check{p}_{D} \in \check{\mathcal{D}}}\sup_{F_{Y,X}\in\mathcal{F}}\mathbb{E}\lVert\check{p}_{D}(\textit{data})-p_{D}^{*}(F_{Y,X})\rVert_{2}^{2}\geq\delta^{2}\left(1-\frac{\sum_{j,j'=1}^{M}\text{KL}(F_{Y,X}^{j}\|F_{Y,X}^{j'})/M^{2}+\log2}{\log M}\right)
\end{align*}
\end{lemma} 
\begin{proof}[Proof of Lemma \ref{lm:fano}]
The result follows from Proposition 15.12 (the Fano's inequality)
and inequality (15.34) (convexity of the KL divergence) in Chapter
15 of \citet{wainwright2019high}, where $\Phi$ is taken to be the
square function, $\rho$ the $L_{2}-$distance, and $\theta$ the
functional $p_{D}^{*}$. 
\end{proof}
\begin{lemma} \label{lm:G_delta-VC-subgraph} Consider the following
function class: 
\begin{align*}
\{(y,x)\mapsto(p\mathbf{1}\{y\geq p\}-\tilde{p}\mathbf{1}\{y\geq\tilde{p}\})\mathbf{1}\{x\in[k/K,(k+1)/K)\}\colon p\in[0,1]\}.
\end{align*}
For any $\tilde{p}\in[0,1]$, $K\geq1$, and $0\leq k\leq K-1$, the
above class is a VC-subgraph with VC-dimension no greater than
$2$. \end{lemma} 
\begin{proof}[Proof of Lemma \ref{lm:G_delta-VC-subgraph}]
By Lemma 2.6.22 in Chapter 2 of \citet{wellner1996}, the class 
\begin{align*}
\{(y,x)\mapsto p\mathbf{1}\{y\geq p\}\colon p\in[0,1]\}
\end{align*}
is a VC-subgraph with VC-dimention no greater than $2$.\footnote{In the original statement of the lemma, the VC dimension is no greater
than 3. This is because the definition of VC dimension in \citet{wellner1996}
is the smallest number $n$ for which no set of $n$ points is shattered.
The definition we use in this paper is the largest number $n$ that
some set of $n$ points is shattered.} The function $(y,x)\mapsto\tilde{p}\mathbf{1}\{y\geq\tilde{p}\}$
is a fixed function that does not depend on the index $p$. By the
proof Lemma 2.6.18(v) in \citet{wellner1996}, the class 
\begin{align*}
\{(y,x)\mapsto p\mathbf{1}\{y\geq p\}-\tilde{p}\mathbf{1}\{y\geq\tilde{p}\}\colon p\in[0,1]\}
\end{align*}
is a VC-subgraph with VC-dimention no greater than $2$. Lastly,
we multiply each function in the class by an indicator $\mathbf{1}\{x\in[k/K,(k+1)/K)\}$.
This does not increase the VC-dimension. 
\end{proof}
\begin{lemma} \label{lm:expected-sup-EP-for-VC-subgraph} Let $Z_{1},\ldots,Z_{n}$
be an i.i.d. sequence of random variables from distribution $P$.
Let $\mathcal{G}$ be a class of VC-subgraph functions with VC-dimension
$v$ and envelope function $G$. Assume that $\lVert G\rVert_{L_{2}(P)}<\infty$.
Then we have 
\begin{align*}
\mathbb{E}\sup_{g\in\mathcal{G}}\left|\frac{1}{n}\sum_{i=1}^{n}g(Z_{i})-\mathbb{E}g(Z_{i})\right|\leq8\sqrt{2}\frac{\lVert G\rVert_{L_{2}(P)}}{\sqrt{n}}\left(\log(2C)+\log(v)+(\log(16)+3)v\right),
\end{align*}
for some universal constant $C$, where the $L_{2}\left(P\right)$
norm $\lVert f-g\rVert_{L_{2}\left(P\right)}\equiv\left(\int_{\mathcal{X}}\left[f\left(x\right)-g\left(x\right)\right]^{2}\mathbb{P}(dx)\right)^{\frac{1}{2}}$.
\end{lemma} 
\begin{proof}[Proof of Lemma \ref{lm:expected-sup-EP-for-VC-subgraph}]
This is a well-known result in the literature. We include it here
for completeness. Let $N(\mathcal{G},L_{2}(Q),\tau)$ denote the covering
number of $(\mathcal{G},L_{2}(Q))$. By Remark 3.5.5 in Chapter 3
of \citet{gine_nickl_2015}, we know that 
\begin{align*}
\mathbb{E}\sup_{g\in\mathcal{G}}\left|\frac{1}{n}\sum_{i=1}^{n}g(X_{i})-\mathbb{E}g(X_{i})\right|\leq8\sqrt{2}\frac{\lVert G\rVert_{L_{2}(P)}}{\sqrt{n}}\int_{0}^{1}\sup_{Q}\sqrt{\log2N(\mathcal{G},L_{2}(Q),\tau\lVert G\rVert_{L_{2}(Q)})}d\tau,
\end{align*}
where the supremum is taken over all discrete probabilities with a
finite number of atoms. By Theorem 2.6.7 in Chapter 2 of \citet{wellner1996},
we know that for any probability measure $Q$, 
\begin{align*}
N(\mathcal{G},L_{2}(Q),\tau\lVert G\rVert_{L_{2}(Q)})\leq Cv(16e)^{v}(1/\tau)^{2v},
\end{align*}
for some universal constant $C$. Therefore, 
\begin{align*}
\int_{0}^{1}\sup_{Q}\sqrt{\log2N(\mathcal{G},L_{2}(Q),\tau\lVert G\rVert_{L_{2}(Q)})}d\tau\leq\log(2C)+\log(v)+(\log(16)+3)v
\end{align*}
Then the desired result follows. 
\end{proof}

\bibliography{robustness}

@article{myerson1981optimal,
  title={Optimal auction design},
  author={Myerson, Roger B},
  journal={Mathematics of operations research},
  volume={6},
  number={1},
  pages={58--73},
  year={1981},
  publisher={INFORMS}
}

@article{fu2021full,
  title={Full surplus extraction from samples},
  author={Fu, Hu and Haghpanah, Nima and Hartline, Jason and Kleinberg, Robert},
  journal={Journal of Economic Theory},
  volume={193},
  pages={105230},
  year={2021},
  publisher={Elsevier}
}

@inproceedings{fu2015randomization,
  title={Randomization beats second price as a prior-independent auction},
  author={Fu, Hu and Immorlica, Nicole and Lucier, Brendan and Strack, Philipp},
  booktitle={Proceedings of the sixteenth ACM conference on economics and computation},
  pages={323--323},
  year={2015}
}

@article{baliga2003market,
  title={Market research and market design},
  author={Baliga, Sandeep and Vohra, Rakesh},
  journal={Advances in theoretical Economics},
  volume={3},
  number={1},
  year={2003},
  publisher={De Gruyter}
}

@inproceedings{zhang2020learning,
  title={Learning the Valuations of a $ k $-demand Agent},
  author={Zhang, Hanrui and Conitzer, Vincent},
  booktitle={International Conference on Machine Learning},
  pages={11066--11075},
  year={2020},
  organization={PMLR}
}

@article{chambers2021recovering,
  title={Recovering preferences from finite data},
  author={Chambers, Christopher P and Echenique, Federico and Lambert, Nicolas S},
  journal={Econometrica},
  volume={89},
  number={4},
  pages={1633--1664},
  year={2021},
  publisher={Wiley Online Library}
}

@article{chambers2023recovering,
  title={Recovering utility},
  author={Chambers, Christopher P and Echenique, Federico and Lambert, Nicolas S},
  journal={arXiv preprint arXiv:2301.11492},
  year={2023}
}

@inproceedings{balcan2014learning,
  title={Learning economic parameters from revealed preferences},
  author={Balcan, Maria-Florina and Daniely, Amit and Mehta, Ruta and Urner, Ruth and Vazirani, Vijay V},
  booktitle={Web and Internet Economics: 10th International Conference, WINE 2014, Beijing, China, December 14-17, 2014. Proceedings 10},
  pages={338--353},
  year={2014},
  organization={Springer}
}

@inproceedings{beigman2006learning,
  title={Learning from revealed preference},
  author={Beigman, Eyal and Vohra, Rakesh},
  booktitle={Proceedings of the 7th ACM Conference on Electronic Commerce},
  pages={36--42},
  year={2006}
}

@article{allouah2022pricing,
  title={Pricing with samples},
  author={Allouah, Amine and Bahamou, Achraf and Besbes, Omar},
  journal={Operations Research},
  volume={70},
  number={2},
  pages={1088--1104},
  year={2022},
  publisher={INFORMS}
}

@article{allouah2023optimal,
  title={Optimal pricing with a single point},
  author={Allouah, Amine and Bahamou, Achraf and Besbes, Omar},
  journal={Management Science},
  year={2023},
  publisher={INFORMS}
}

@article{basu2020falsifiability,
  title={On the falsifiability and learnability of decision theories},
  author={Basu, Pathikrit and Echenique, Federico},
  journal={Theoretical Economics},
  volume={15},
  number={4},
  pages={1279--1305},
  year={2020},
  publisher={Wiley Online Library}
}

@article{gonccalves2020statistical,
      title={Statistical Mechanism Design: Robust Pricing, Estimation, and Inference}, 
      author={Duarte Gonçalves and Bruno A. Furtado},
      journal={arXiv preprint arXiv:2405.17178},
      year={2024},
      url={https://arxiv.org/abs/2405.17178}, 
}

@inproceedings{babaioff2018two,
  title={Are two (samples) really better than one?},
  author={Babaioff, Moshe and Gonczarowski, Yannai A and Mansour, Yishay and Moran, Shay},
  booktitle={Proceedings of the 2018 ACM Conference on Economics and Computation},
  pages={175--175},
  year={2018}
}

@inproceedings{devanur2016sample,
  title={The sample complexity of auctions with side information},
  author={Devanur, Nikhil R and Huang, Zhiyi and Psomas, Christos-Alexandros},
  booktitle={Proceedings of the forty-eighth annual ACM symposium on Theory of Computing},
  pages={426--439},
  year={2016}
}

@book{jank2010modeling,
  title     = {Modeling online auctions},
  author    = {Jank, Wolfgang and Shmueli, Galit},
  year      = {2010},
  publisher = {John Wiley \& Sons}
}

@article{cremer1988full,
  title     = {Full extraction of the surplus in Bayesian and dominant strategy auctions},
  author    = {Cr{\'e}mer, Jacques and McLean, Richard P},
  journal   = {Econometrica: Journal of the Econometric Society},
  pages     = {1247--1257},
  year      = {1988},
  publisher = {JSTOR}
}

@article{BergemannSchlag2008,
  title = {Pricing without {{Priors}}},
  author = {Bergemann, Dirk and Schlag, Karl H.},
  year = {2008},
  journal = {Journal of the European Economic Association},
  shortjournal = {Journal of the European Economic Association},
  volume = {6},
  number = {2-3},
  pages = {560--569},
  issn = {1542-4766, 1542-4774},
  doi = {10.1162/JEEA.2008.6.2-3.560},
  url = {https://academic.oup.com/jeea/article-lookup/doi/10.1162/JEEA.2008.6.2-3.560},
  urldate = {2022-08-23},
  langid = {english}
}

@article{BergemannSchlag2011,
  title = {Robust Monopoly Pricing},
  author = {Bergemann, Dirk and Schlag, Karl H.},
  year = {2011},
  journal = {Journal of Economic Theory},
  shortjournal = {Journal of Economic Theory},
  volume = {146},
  number = {6},
  pages = {2527--2543},
  issn = {00220531},
  doi = {10.1016/j.jet.2011.10.018},
  url = {https://linkinghub.elsevier.com/retrieve/pii/S0022053111001517},
  urldate = {2022-08-02},
  langid = {english}
}

@article{carrollRobustnessMechanismDesign2019,
  title = {Robustness in {{Mechanism Design}} and {{Contracting}}},
  author = {Carroll, Gabriel},
  year = {2019},
  journal = {Annual Review of Economics},
  shortjournal = {Annu. Rev. Econ.},
  volume = {11},
  number = {1},
  pages = {139--166},
  issn = {1941-1383, 1941-1391},
  doi = {10.1146/annurev-economics-080218-025616},
  url = {https://www.annualreviews.org/doi/10.1146/annurev-economics-080218-025616},
  urldate = {2021-07-21},
  langid = {english}
}

@article{huang2018making,
	title={Making the most of your samples},
	author={Huang, Zhiyi and Mansour, Yishay and Roughgarden, Tim},
	journal={SIAM Journal on Computing},
	volume={47},
	number={3},
	pages={651--674},
	year={2018},
	publisher={SIAM}
}

@article{dhangwatnotai2015revenue,
	title={Revenue maximization with a single sample},
	author={Dhangwatnotai, Peerapong and Roughgarden, Tim and Yan, Qiqi},
	journal={Games and Economic Behavior},
	volume={91},
	pages={318--333},
	year={2015},
	publisher={Elsevier}
}

@book{wellner1996,
	title     = "Weak Convergence and Empirical Processes",
	author={van der Vaart, Aad W and Wellner, Jon A},
	year      = 1996,
	publisher = "Springer",
	address   = "New York, NY",
	isbn = "978-1-4757-2545-2"
}

@article{talagrand1996new,
	title={New concentration inequalities in product spaces},
	author={Talagrand, Michel},
	journal={Inventiones mathematicae},
	volume={126},
	number={3},
	pages={505--563},
	year={1996},
	publisher={Springer}
}

@incollection{bousquet2003concentration,
	title={Concentration inequalities for sub-additive functions using the entropy method},
	author={Bousquet, Olivier},
	booktitle={Stochastic inequalities and applications},
	pages={213--247},
	year={2003},
	publisher={Springer}
}

@book{wainwright2019high,
	title={High-dimensional statistics: A non-asymptotic viewpoint},
	author={Wainwright, Martin J},
	volume={48},
	year={2019},
	publisher={Cambridge University Press}
}

@book{folland1999real,
	title={Real analysis: modern techniques and their applications},
	author={Folland, Gerald B},
	volume={40},
	year={1999},
	publisher={John Wiley \& Sons},
	edition        = {2},
	address = {New York, NY},
}

@book{gine_nickl_2015, 
	place={Cambridge}, 
	series={Cambridge Series in Statistical and Probabilistic Mathematics}, 
	title={Mathematical Foundations of Infinite-Dimensional Statistical Models}, 
	DOI={10.1017/CBO9781107337862}, 
	publisher={Cambridge University Press}, 
	author={Gin\'e, Evarist and Nickl, Richard}, 
	year={2015}, 
	collection={Cambridge Series in Statistical and Probabilistic Mathematics}
}

@book{tsybakov2009introduction,
	title={Introduction to Nonparametric Estimation},
	author={Tsybakov, Alexandre B.},
	year={2009},
	publisher={Springer, New York, NY},
	edition        = {1},
	series={Springer Series in Statistics},
}

@book{cover2005elements,
	
	publisher = {John Wiley \& Sons, Ltd},
	isbn = {9780471748823},
	author={Thomas M. Cover and Joy A. Thomas},
	title = {Elements of Information Theory},
	doi = {https://doi.org/10.1002/047174882X.ch2},
	url = {https://onlinelibrary.wiley.com/doi/abs/10.1002/047174882X.ch2},
	eprint = {https://onlinelibrary.wiley.com/doi/pdf/10.1002/047174882X.ch2},
	year = {2005},
}

@inproceedings{guo2019settling,
	title={Settling the sample complexity of single-parameter revenue maximization},
	author={Guo, Chenghao and Huang, Zhiyi and Zhang, Xinzhi},
	booktitle={Proceedings of the 51st Annual ACM SIGACT Symposium on Theory of Computing},
	pages={662--673},
	year={2019}
}

@inproceedings{cole2014sample,
	title={The sample complexity of revenue maximization},
	author={Cole, Richard and Roughgarden, Tim},
	booktitle={Proceedings of the forty-sixth annual ACM symposium on Theory of computing},
	pages={243--252},
	year={2014}
}

@article{segal2003optimal,
	title={Optimal pricing mechanisms with unknown demand},
	author={Segal, Ilya},
	journal={American Economic Review},
	volume={93},
	number={3},
	pages={509--529},
	year={2003}
}

@article{goldberg2006competitive,
	title={Competitive auctions},
	author={Goldberg, Andrew V and Hartline, Jason D and Karlin, Anna R and Saks, Michael and Wright, Andrew},
	journal={Games and Economic Behavior},
	volume={55},
	number={2},
	pages={242--269},
	year={2006},
	publisher={Elsevier}
}

@article{hartmann2011identifying,
  title={Identifying causal marketing mix effects using a regression discontinuity design},
  author={Hartmann, Wesley and Nair, Harikesh S and Narayanan, Sridhar},
  journal={Marketing Science},
  volume={30},
  number={6},
  pages={1079--1097},
  year={2011},
  publisher={INFORMS}
}

@article{hirano2003asymptotic,
  title={Asymptotic efficiency in parametric structural models with parameter-dependent support},
  author={Hirano, Keisuke and Porter, Jack R},
  journal={Econometrica},
  volume={71},
  number={5},
  pages={1307--1338},
  year={2003},
  publisher={Wiley Online Library}
}

@article{alexander1987rates,
  title={Rates of growth and sample moduli for weighted empirical processes indexed by sets},
  author={Alexander, Kenneth S},
  journal={Probability Theory and Related Fields},
  volume={75},
  number={3},
  pages={379--423},
  year={1987},
  publisher={Springer}
}

@book{van_de_geer2000empirical,
  title={Empirical Processes in M-estimation},
  author={van de Geer, Sara A},
  volume={6},
  year={2000},
  publisher={Cambridge university press}
}

@article{kolmogorov1959varepsilon,
  title={$\varepsilon$-entropy and $\varepsilon$-capacity of sets in function spaces},
  author={Kolmogorov, Andrei Nikolaevich and Tikhomirov, Vladimir Mikhailovich},
  journal={Uspekhi Matematicheskikh Nauk},
  volume={14},
  number={2},
  pages={3--86},
  year={1959},
  publisher={Russian Academy of Sciences, Steklov Mathematical Institute of Russian~…}
}
\bibliographystyle{aer}
\end{document}